\documentclass[pdflatex,sn-basic,iicol]{sn-jnl}

\usepackage{graphicx}%
\usepackage{multirow}%
\usepackage{amsmath,amssymb,amsfonts}%
\usepackage{amsthm}%
\usepackage{mathrsfs}%
\usepackage[title]{appendix}%
\usepackage{xcolor}%
\usepackage{textcomp}%
\usepackage{manyfoot}%
\usepackage{booktabs}%
\usepackage{algorithm}%
\usepackage{algorithmicx}%
\usepackage{algpseudocode}%
\usepackage{listings}%

\usepackage{dsfont}
\usepackage{makecell}

\usepackage{tikz,hyperref}
\usepackage{caption}

\usepackage{soul} 

\definecolor{lime}{HTML}{A6CE39}
\DeclareRobustCommand{\orcidicon}{%
	\begin{tikzpicture}
	\draw[lime, fill=lime] (0,0) 
	circle [radius=0.16] 
	node[white] {{\fontfamily{qag}\selectfont \tiny ID}};
	\draw[white, fill=white] (-0.0625,0.095) 
	circle [radius=0.007];
	\end{tikzpicture}
	\hspace{-2mm}
}
 
\foreach \x in {A, ..., Z}{%
	\expandafter\xdef\csname orcid\x\endcsname{\noexpand\href{https://orcid.org/\csname orcidauthor\x\endcsname}{\noexpand\orcidicon}}
}

\theoremstyle{thmstyleone}%
\newtheorem{theorem}{Theorem}
\newtheorem{lemma}{Lemma}
\newtheorem{assumption}{Assumption}

\theoremstyle{thmstyletwo}%

\theoremstyle{thmstylethree}%

\begin{document}

\title[Article Title]{FedDyMem: Efficient Federated Learning with Dynamic Memory and Memory-Reduce for Unsupervised Image Anomaly Detection}

\author[1]{\fnm{Silin} \sur{Chen}\orcidA{}}\email{silin.chen@smail.nju.edu.cn}
\author[1]{\fnm{Andy} \sur{Liu}\orcidB{}}\email{andyliu@smail.nju.edu.cn}
\author[1]{\fnm{Kangjian} \sur{Di}\orcidC{}}\email{kangjiandi@smail.nju.edu.cn}
\author[2,3]{\fnm{Yichu} \sur{Xu}}\email{xuyc@lamda.nju.edu.cn}
\author[2,3]{\fnm{Han-Jia} \sur{Ye}\orcidD{}}\email{yehj@lamda.nju.edu.cn}
\author[4]{\fnm{Wenhan} \sur{Luo}\orcidE{}}\email{whluo.china@gmail.com}
\author*[1,5]{\fnm{Ningmu} \sur{Zou}\orcidF{}}\email{nzou@nju.edu.cn}
\affil[1]{\orgdiv{School of Integrated Circuits}, \orgname{Nanjing University}, \orgaddress{\city{Suzhou}, \postcode{215163}, \state{Jiangsu}, \country{China}}}

\affil[2]{\orgdiv{National Key Laboratory for Novel Software Technology}, \orgname{Nanjing University}, \orgaddress{\city{Nanjing}, \postcode{210023}, \state{Jiangsu}, \country{China}}}

\affil[3]{\orgdiv{School of Artificial Intelligence}, \orgname{Nanjing University}, \orgaddress{\city{Nanjing}, \postcode{210023}, \state{Jiangsu}, \country{China}}}

\affil[4]{\orgname{Hong Kong University of Science and Technology}, \orgaddress{\state{Hong Kong}}}
\affil[5]{\orgdiv{Interdisciplinary Research Center for Future Intelligent Chips (Chip-X)}, \orgname{Nanjing University}, \orgaddress{\city{Suzhou}, \postcode{215163}, \state{Jiangsu}, \country{China}}}


\abstract{Unsupervised image anomaly detection (UAD) has become a critical process in industrial and medical applications, where models are trained on normal samples without requiring labeled anomalies. However, even normal images may contain sensitive product designs, manufacturing information, or institution-specific population characteristics, making centralized data collection difficult in practical deployments. Federated UAD provides a promising decentralized solution, but it remains challenging because clients use only normal training samples, their feature distributions remain highly heterogeneous within the same task, and communication should be efficient while limiting the exposure of sensitive client information. To address these challenges, we propose FedDyMem, an efficient federated learning method with dynamic memory and memory-reduce for unsupervised image anomaly detection. Instead of exchanging full model parameters, FedDyMem facilitates knowledge sharing between the clients and server through reduced dynamic memory banks. In the local clients, a memory generator and a metric loss are employed to improve the consistency of the feature distribution for normal samples, leveraging the local model to update the memory bank dynamically. For efficient communication and data privacy, a memory-reduce method based on weighted averages is proposed to significantly decrease the scale of memory banks. This representation reduces communication overhead and empirically lowers reconstruction fidelity under inversion attacks. On the server, global memory is constructed and distributed to individual clients through k-means aggregation. Experiments on six federated UAD tasks constructed from eleven public industrial and medical datasets show that FedDyMem consistently improves anomaly detection and localization while reducing communication cost.}

\keywords{Unsupervised federated learning, Unsupervised image anomaly detection, Feature distribution shift, Communication efficiency}



\maketitle

\section{Introduction}\label{sec1}
Unsupervised image anomaly detection has achieved significant success in various domains, such as industrial inspection (\cite{he2024mambaad}) and medical disease recognition (\cite{cai2024medianomaly}). Obtaining reliable annotations for abnormal samples can be time-consuming. UAD therefore builds anomaly detectors from normal training data without requiring annotations of real abnormal samples. However, existing methods typically rely on centralized training with large-scale datasets composed of normal samples. In practice, normal data in industrial and medical scenarios remain highly sensitive. For industrial inspection, normal images may contain critical intellectual property, such as PCB layouts, wafer patterns, component geometries, and material textures (\cite{fyrbiak2017hardware}). Moreover, these images may implicitly reveal the deployed automated optical inspection (AOI) solutions, inspection logic, and acceptance criteria (\cite{cao2025aoi}). For medical screening, private information related to patient attributes, imaging protocols, and institutional screening populations may still be exposed even from normal images (\cite{macpherson2023patient}). Consequently, industrial companies and healthcare organizations often have practical limitations in collecting and aggregating raw data, which significantly challenge centralized learning approaches (\cite{zhang2024vertical,zhao2024medical}). 

\begin{figure}[t]
\centering
\includegraphics[width=0.5\textwidth]{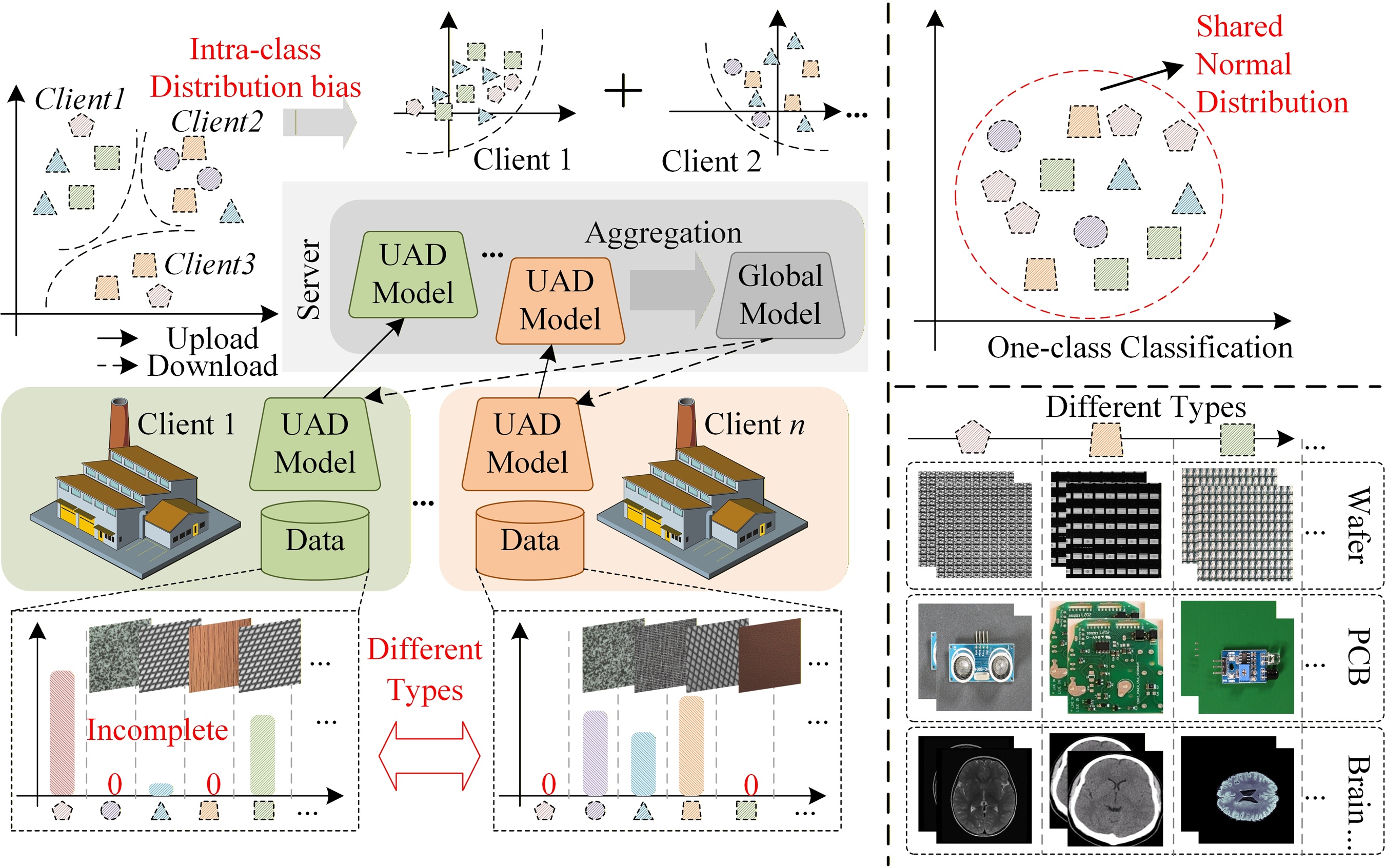} 
\caption{Illustration of federated learning in UAD. As shown on the left, for the same inspection or screening task, an intra-class distribution bias exists between local models, caused by clients possessing different and incomplete subsets of normal data within the same inspection or screening task. As shown in the top right, federated UAD aims to construct a shared normal distribution for one-class anomaly detection across clients. Examples of type differences in anomaly detection are shown in the bottom right.}
\label{fig_1}
\end{figure}

Recently, federated learning has been developed as a privacy-preserving paradigm for machine learning, providing collaborative model training across distributed devices while maintaining raw data locally and not transmitting to central servers (\cite{saha2024multifaceted}). Thus, the integration of local training and global aggregation strategies has garnered significant attention in industrial and healthcare domains, enabling data privacy protection (\cite{zhao2024medical,  zhang2024vertical}). We consider federated UAD among clients that share the same inspection or screening task. For example, separate federations may be formed for wafer defect detection, PCB inspection, or chest X-ray screening, rather than combining semantically unrelated anomaly detection tasks within the same federation. Each client uses only local normal samples for training, without sharing raw data. Moreover, directly averaging model parameters across heterogeneous clients may not preserve the normal representations required for anomaly scoring. Thus, federated UAD involves not only distributed optimization but also the collaborative learning of normal representations across clients. We summarize the main challenges as follows:

\noindent\textbf{1) End-to-end unsupervised anomaly detection without additional labeled data or supervised fine-tuning.} Most existing unsupervised federated learning methods focus on self-supervised representation learning or prototype learning, and typically require local fine-tuning through a supervised linear evaluation protocol after federated training (\cite{fedU, liao2024rethinking, fedX, zhang2024prototype}). This limits their applicability to UAD, where only normal samples are used during training. 

\noindent\textbf{2) Heterogeneous normal distributions across clients.} As shown in Fig.~\ref{fig_1}, although all clients address the same UAD task and use only normal samples for training, their feature distributions can vary substantially across clients due to differences in product subtypes, manufacturing conditions, imaging devices, acquisition protocols, or screening populations (\cite{lei2025adapted}). Such heterogeneous normal modes may also arise in centralized UAD when diverse normal data are modeled jointly (\cite{you2022unified}). 
The additional difficulty in the federated setting is that these normal modes are not jointly observable at the sample level during optimization. Consequently, local objectives may become specialized to incomplete normal supports. This limited local coverage may cause a model to assign high anomaly scores to normal images from other clients. Inconsistencies among local updates can hinder the construction of a federation-wide normal representation, leading to performance degradation and convergence difficulties under data heterogeneity (\cite{li2020federated, karimireddy2020scaffold, lei2025adapted}).

\noindent\textbf{3) Communication-efficient and privacy-aware knowledge sharing.} Standard federated learning methods commonly exchange full model parameters or gradients, which may incur heavy communication overhead in resource-constrained industrial and medical scenarios (\cite{li2020federated, almanifi2023communication}). More importantly, the exchanged information may still encode product textures, imaging protocols, or patient-related distributional patterns, making it sensitive to potential reconstruction or inversion attacks (\cite{cao2025hybrid}).

To address these issues, we introduce FedDyMem, an efficient federated learning framework with dynamic memory and memory-reduce, specifically designed for federated UAD. FedDyMem facilitates knowledge transfer in federated learning by sharing the memory bank instead of model parameters among clients. Specifically, it constructs a local memory bank using the limited dataset of normal samples available at each client. To address the challenge of intra-class distribution bias in different clients, a memory generator and a metric loss function are introduced to improve the consistency of normal feature distributions. Then the local memory bank is dynamically updated by the local client model before uploading. To mitigate distribution bias during aggregation, FedDyMem employs k-means clustering to extract general distribution from different clients, thereby ensuring effective aggregation without confusion. Considering communication efficiency and data privacy, a memory-reduce method is proposed to decrease the scale of memory banks, thereby reducing its mutual information with the original data during dynamic updates. We then evaluate our proposed method separately on six federated UAD tasks corresponding to different product or health screening types, achieving state-of-the-art performance in each case. The main contributions are summarized as follows:
\begin{itemize} 
\item To the best of our knowledge, FedDyMem is the first federated learning framework designed specifically for unsupervised image anomaly detection. This framework facilitates collaborative training using only normal samples for local training.  
\item To address the significant distribution biases of normal samples across clients, we incorporate a memory generator and a metric loss function that improve the consistency of feature distributions. Additionally, we employ k-means clustering during the aggregation phase to reduce inter-client discrepancies and mitigate potential ambiguities.
\item To achieve both communication efficiency and privacy-aware knowledge sharing, we propose a memory-reduce method based on a dynamic weighted average. This method substantially decreases the size of the memory bank, thereby reducing both communication overhead and the mutual information between uploaded information and raw data, while preserving model performance.
\item We construct six federated UAD tasks by grouping eleven public image anomaly detection datasets according to product or health screening types. Comprehensive experiments demonstrate that FedDyMem achieves excellent federated UAD performance.
\end{itemize}

\section{Related Work}
\subsection{Unsupervised Image Anomaly Detection}
The objective of UAD is to identify whether a given sample is anomalous and to precisely localize the anomaly regions, using a training dataset that comprises only normal samples (\cite{xie2024iad}). The reconstruction-based methods are significant in UAD, which assumes that anomalous samples cannot be accurately reconstructed by feature learning models trained exclusively on normal samples (\cite{liang2023omni}). In contrast to traditional UAD approaches, which rely on a single, centralized training dataset, federated UAD requires the efficient aggregation of information from samples distributed across multiple clients. Reconstruction-based methods, however, typically rely solely on normal samples within the training set and fail to capture key characteristics of out-of-distribution samples. When these centralized UAD objectives are directly optimized on each client in a federated setting, the learned representations may fit client-specific normal modes rather than the shared normal distribution required for federated UAD (\cite{you2022unified}). Another widely utilized approach for anomaly detection is the memory bank-based method. SPADE (\cite{cohen2020sub}) proposed a semantic pyramid structure to construct a pixel-level feature memory, effectively facilitating anomaly detection. PaDiM (\cite{defard2021padim}) employed Gaussian distributions derived from normal samples as the memory bank and utilized the Mahalanobis distance as the anomaly metric. PatchCore (\cite{roth2022towards}) introduced local neighborhood aggregation to expand the receptive field of the memory bank while preserving resolution. Additionally, PatchCore implemented a greedy core-set subsampling strategy to reduce the memory bank size without significant performance degradation. PNI (\cite{bae2023pni}) further advanced the paradigm by integrating spatial and neighborhood information into the memory bank construction, thereby improving its capacity to represent normal samples comprehensively. These methods extract features from normal images and store them within a feature memory bank. During testing, the sample queries the memory bank to retrieve features corresponding to the k-nearest neighbors. However, the memory banks are non-trainable, limiting their capacity to learn information across local clients in federated learning. Beyond these category-specific approaches, unified UAD aims to detect anomalies from multiple semantic categories using a single model. UniAD (\cite{you2022unified}) introduced a unified framework for multi-class anomaly detection, enabling heterogeneous categories to be jointly modeled by a single anomaly detector. MambaAD (\cite{he2024mambaad}) further explored state-space models for multi-class UAD to improve the modeling of heterogeneous normal patterns. More recently, InpFormer (\cite{luo2025exploring}) investigated intrinsic normal prototypes within individual images for universal anomaly detection, further enhancing the representation of diverse normal patterns. These studies suggest that appropriately designed unified models can mitigate inter-category interference and effectively model heterogeneous normal distributions when data from different categories are jointly available during training. However, these methods generally assume centralized access to heterogeneous normal data, whereas each client in federated UAD observes only its local normal distribution. Therefore, learning a unified normal representation requires integrating complementary normal modes across clients through federated communication rather than direct joint training.
\subsection{Federated Learning}
\noindent\textbf{Federated Learning with Data Heterogeneity.} FedAvg (\cite{fedavg}) is a foundational approach in federated learning that enables the training of a global model by aggregating parameters from locally trained models. Despite its effectiveness, the performance of FedAvg degrades significantly with high data heterogeneity. FedProx (\cite{li2020federated2}) mitigates client drift by introducing a proximal regularization term, while SCAFFOLD (\cite{karimireddy2020scaffold}) employs control variates to correct local update directions, thereby improving optimization under data heterogeneity. Other methods address data heterogeneity through personalized models. FedBN (\cite{li2021fedbn}) mitigates data heterogeneity challenges in federated learning by maintaining BN parameters that are specific to each client's local model. APPLE (\cite{luo2022adapt}) aggregates client models locally by learning precise weight updates instead of relying on approximations. However, these methods are mainly designed for supervised federated learning and are not directly tailored to federated UAD. In UAD, all training samples belong to the normal class, and directly combining local UAD objectives with parameter aggregation may cause each client to overfit its own normal distribution. Recent studies address this issue by sharing or synthesizing auxiliary data across clients (\cite{nguyen2021federated, li2024feature}). Although such strategies can improve different client distribution coverage, sharing additional data among clients may increase the risk of data privacy leakage (\cite{li2024fedcir}). More recently, prototype-based federated learning methods have been developed to reduce generalization gaps by aligning global and local feature prototypes (\cite{tan2022fedproto, feddbe}). These methods are effective when semantic class labels are available, because category-specific prototypes can serve as explicit anchors for cross-client feature alignment. However, federated UAD only contains normal samples during training, making it difficult to construct discriminative category-specific prototypes.

\noindent\textbf{Unsupervised Federated Learning.}
In unsupervised federated learning, most existing methods focus on self-supervised representation learning, knowledge distillation, or prototype-based representation learning, which aim to learn good feature representations from unlabeled data to facilitate downstream machine learning tasks (\cite{fedU, liao2024rethinking, fedX, zhang2024prototype}). However, these methods depend on aligning local representations with global representations, a task that proves challenging for UADs constrained to a single category. ProtoFL (\cite{kim2023protofl}) was proposed to integrate federated learning with prototyping for addressing one-class classification. However, its two-stage fine-tuning method, which utilizes normalizing flow for representation learning, introduces additional computational overhead and ignores distribution bias among clients in UAD. Furthermore, many unsupervised federated learning approaches require supervised fine-tuning, which limits their applicability to UAD. 

\noindent\textbf{Communication-Efficient and Privacy-Aware Federated Learning.}
Communication between clients in federated learning is limited by network bandwidth, unreliable connections, and varying device statuses. Consequently, designing communication-efficient algorithms is crucial to overcome these challenges (\cite{almanifi2023communication}). Existing methods primarily tackle this issue by either decreasing the total number of communication rounds required for model convergence or reducing the data uploaded during each communication round (\cite{li2020federated}). To accelerate global convergence, some approaches leverage mechanisms such as localized model updates (\cite{li2020federated2}) and gradient optimization enhancements (\cite{yang2022federated}), effectively reducing the number of communication rounds. To reduce the data uploaded, methods such as model compression (\cite{shlezinger2020uveqfed, li2023online}) and selective client participation (\cite{xu2020client}) have been extensively explored to minimize communication overhead. Beyond efficiency, privacy-aware communication is equally essential, especially in resource-constrained and privacy-sensitive environments. The exchanged information should resist reconstruction attempts from adversaries, preventing leakage of the raw client data (\cite{cao2025hybrid}). These studies suggest that communication reduction and privacy-aware representation sharing are both important, yet existing methods do not specifically address this for federated UAD with heterogeneous normal distributions.

\section{Preliminary}
\subsection{Typical Federated Learning}
In the typical federated learning, such as FedAvg (\cite{fedavg}), there exist $N$ clients $\{{\mathcal C}_n\}_{n=1}^N$ with their individual datasets ${\{{\mathcal D}_n= ({\mathcal X}_{n}, {\mathcal Y}_{n}) \sim {\mathbb P}_n({\mathcal X}, {\mathcal Y})\}}_{n=1}^{N}$ and local models $\{f_{n}(\cdot;\theta_{n})\colon {{\mathcal X}_{n} \to {{\mathcal Y}_n}}\}_{n=1}^{N}$, where $n$ is the index for the client, $({\mathcal X}_{n}, {\mathcal Y}_n)$ denotes a set of samples and labels on the $n$-th client, ${{{\mathbb P}_{n}}}({\mathcal X}, {\mathcal Y})$ represent the joint distribution of samples and labels. The objective of the typical federated learning is to facilitate collaborative training of a global model $\theta_{\mathrm{glob}}$ that generalizes effectively across the entire dataset $\mathcal D$ while ensuring data privacy. The global training optimization target is:
\begin{equation}
\label{tyfed}
\underset {\theta_{\mathrm{glob}}}{\min {\mathcal L}(\theta_{\mathrm{glob}})} \buildrel \Delta \over = \sum\limits_{n = 1}^N {z_n{{\mathcal L}_n}(\theta_{n})},
\end{equation}
where the $z_n$ denotes the weight coefficient which is computed as $\frac{{\left| {{\mathcal{D}}_n} \right|}}{{\left| \mathcal{D} \right|}}$ in FedAvg, where the $\mathcal{D} = {\bigcup\nolimits_{n = 1}^N {{{\mathcal D}_n}}}$. The empirical loss  ${{\mathcal L}_n}(\theta_{n})$ of the $n$-th client can be formulated as:
\begin{equation}
    {{\mathcal{L}}_n}(\theta_{n}) = {{\mathbb E}_{({x \in {{\mathcal X}_n}},{y \in {{\mathcal Y}_n}}) \sim {\mathbb{P}}_n({\mathcal X}, {\mathcal Y})}}[{{\ell}({f_{n}({x};\theta_{n})}, {y})}],
\end{equation}
where $\ell(\cdot,\cdot)$ denotes the loss function on the clients. Following the local updates, the server $\mathcal{S}$ performs aggregation of multiple local models' parameters to derive the global model as $\theta_{\mathrm{glob}}\leftarrow\sum\nolimits_{n = 1}^N {{z_n}{\theta_{n}}}$.

\subsection{Problem Formulation in Federated UAD}
As shown in Fig.~\ref{fig_1}, we consider a federated UAD setting with $N$ independent clients that collaboratively perform a shared anomaly detection task using only normal samples during training and without sharing raw data. In this work, a client denotes an independent data holder, such as a factory site, production line, or screening center. The shared task refers to anomaly detection for the same inspection or screening target, e.g., PCB inspection, wafer defect detection, or chest X-ray screening. Consequently, the training dataset available at the $n$-th client is denoted as:
 \begin{equation}
     {{\mathcal D}_n^{\mathrm{train}}=({\mathcal X}_{n}, {\mathcal Y}_{n})\sim{\mathbb{P}}_n({\mathcal X})}, \forall{y \in {\mathcal Y}_{n}}, y=0,
 \end{equation}
where $y=0$ indicates that only normal samples are used for training. Although all the ${\mathcal D}_n^{\mathrm{train}}$ are from the same anomaly detection task, normal data distribution in each client can be different due to manufacturing conditions or screening populations. We refer to this normal distribution heterogeneity as intra-class distribution bias in federated UAD, which can be formulated as:
\begin{equation}
\label{skew}
\exists_{x_m \in {\mathcal D}_{m}^{\mathrm{train}}, x_n \in {\mathcal D}_{n}^{\mathrm{train}}}\mathbb{P}_m(x_m) \neq \mathbb{P}_n(x_n), \textit{while } m \neq n.
\end{equation}

During testing, we evaluate the anomaly detection performance across all participating clients under the same federated UAD task. This evaluates whether the learned knowledge can generalize across heterogeneous client distributions within the shared task. Thus, the overall test dataset is defined as:
\begin{gather}
    {\mathcal D}^{\mathrm{test}} = \bigcup\nolimits_{n = 1}^N {{{\mathcal D}_n^{\mathrm{test}}}}, \textit{where} \\
{{\mathcal D}_n^{\mathrm{test}} = ({\mathcal X}_{n}, {\mathcal Y}_{n})\sim{\mathbb{P}}_n({\mathcal X},{\mathcal Y})}, \exists{y \in {\mathcal Y}_{n}}, y=1, \notag
\end{gather}
where $y=1$ denotes the anomalous samples in testing.

Different from supervised federated learning, federated UAD cannot directly formulate the local objective as supervised prediction from  ${\mathcal X}_{n} \to {{\mathcal Y}_n}$. Thus, the empirical loss of the $n$-th client in UAD is formulated as:
\begin{equation}
    {{\mathcal{L}}_n}(\theta_n) = {{\mathbb E}_{({x \in {{\mathcal X}_n}}) \sim {\mathbb{P}}_n({\mathcal X})}}[{{\ell}({f_{n}({x};\theta_n)})}],
\end{equation}
where $\ell(\cdot)$ denotes the local loss function for individual clients. It requires the $n$-th local model $f_n(\cdot;\theta_n)$ to approximate the normal sample distribution, i.e., ${\mathbb{P}_{n}({\mathcal X}^{\mathrm{train}})}$ , under the condition of intra-class bias (refer to \ref{skew}). For the federated UAD, shown in Fig.~\ref{fig_1}, the federated objective in one-class classification task is to generalize across all clients by constructing ${\mathbb{P}({\{{\mathcal X}_{n}^{\mathrm{train}}}\}_{n=1}^{N})}$. During the testing phase, the result of a sample $x^{\mathrm{test}} \in {{\mathcal D}^{\mathrm{test}}}$ can be expressed as follows:
\begin{equation}
\mathds{1}({x^{\mathrm{test}}}) =  \begin{cases}
0,&{\text{if}}\ {{f({x^{\mathrm{test}}};\Omega)}\sim {\mathbb{P}({\{{\mathcal X}_{n}^{\mathrm{train}}}\}_{n=1}^{N})}} \\ 
{1,}&{\text{otherwise,}} 
\end{cases}
\end{equation}
where $\Omega$ denotes the shared knowledge obtained through federated aggregation. For parameter-aggregation methods, $\Omega$ corresponds to the aggregated global parameter $\theta_{\mathrm{glob}}$. In FedDyMem, $\Omega$ denotes the global memory.

\begin{figure*}[ht]
\centering
\includegraphics[width=0.9\textwidth]{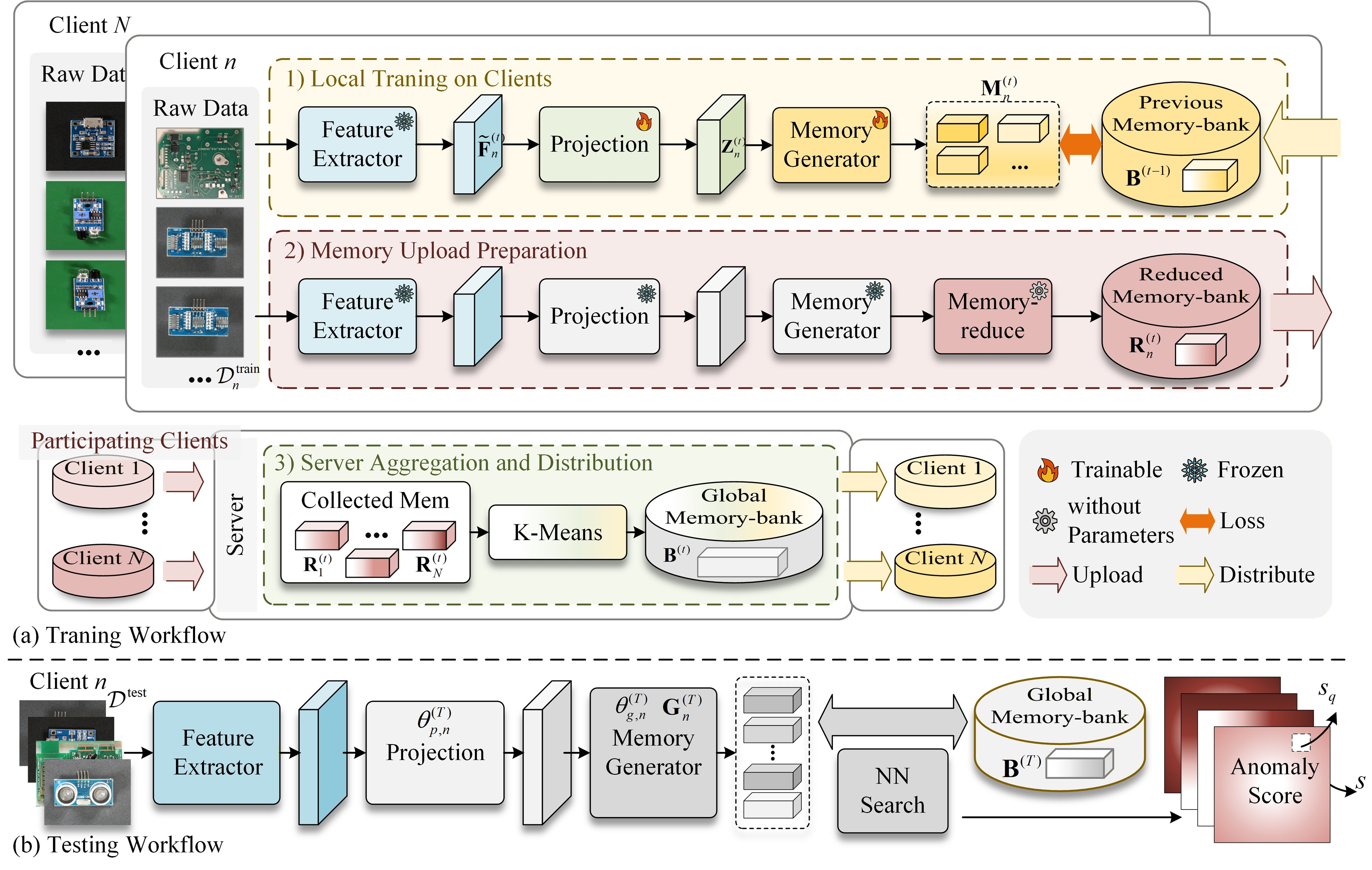}
\caption{Overview of our proposed FedDyMem framework. 
}
\label{fig_2}
\end{figure*}

\section{Methodology}
\subsection{Overview}

As shown in Fig.~\ref{fig_2}, we propose an efficient federated learning framework with dynamic memory and memory-reduce for UAD, called FedDyMem. FedDyMem is designed for a shared federated UAD task, where participating clients use only normal samples from the same inspection or screening target but exhibit heterogeneous normal distributions. Instead of sharing raw data or full model parameters, FedDyMem performs collaborative anomaly detection by exchanging reduced memory banks between clients and the server. Before iterative training, an initial global memory bank is constructed by reducing memory features from clients and aggregating them on the server, which provides the reference memory for the first trainable round. In each communication round, FedDyMem follows three steps:

\noindent\textbf{Step 1: Local training on clients.}
Each client receives the global memory bank from the previous round and updates its local projection layer and memory generator using only normal samples. The local objective aligns generated memory features with the shared global memory.

\noindent\textbf{Step 2: Memory upload preparation.}
After local training, each client regenerates memory features from its normal samples and compresses them into a compact local memory bank through memory-reduce. Only this reduced memory bank is uploaded to the server.

\noindent\textbf{Step 3: Server aggregation and distribution.}
The server aggregates the uploaded reduced memory banks into a new global memory bank and distributes it back to clients for the next round.

After training, each client performs anomaly detection using its locally trained modules and the final global memory bank. Following the anomaly score function in (\cite{roth2022towards, lee2022cfa}), test memory features are compared with the global memory through nearest-neighbor search to compute anomaly scores.

\begin{figure}[t]
\centering
\includegraphics[width=0.45\textwidth]{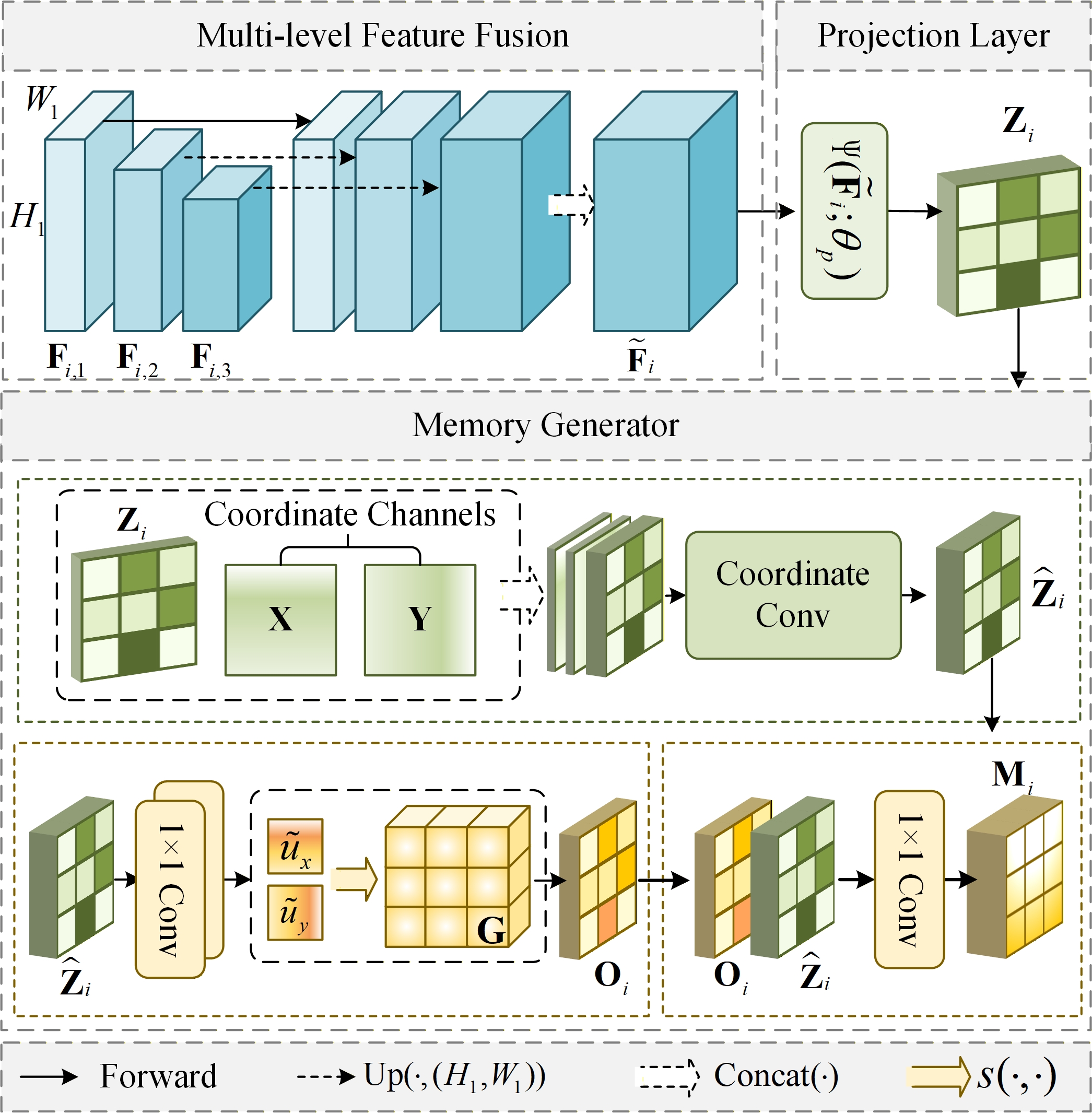}
\caption{
Illustration of the memory feature generation process in each client model.
}
\label{fig_3}
\end{figure}

\subsection{Model Architecture in Clients}
The client model consists of a frozen backbone, a projection layer, and a memory generator. Given an input image $x_{n,i}$, the client model extracts multi-level pretrained features and adapts them through the projection layer to generate a per-sample memory feature $\mathbf{M}_{n,i}$.

\noindent{\textbf{Frozen Backbone.}} Locally, FedDyMem employs a frozen CNN model $\Phi(\cdot)$ pre-trained on ImageNet to extract multi-level hierarchical features. Motivated by (\cite{roth2022towards}), FedDyMem performs multi-level feature fusion on the extracted features. As shown in Fig.~\ref{fig_3}, for a given sample $x_{n,i} \in {\mathbb R}^{H \times W \times 3}$ which is from the $n$-th client's ${\mathcal D}_{n}$, the extracted features are hierarchically organized into $L$ layers (e.g., $L=5$ in the case of ResNet (\cite{he2016deep})). We define ${\mathbf{F}_{n,i,l}} \in {\mathbb R}^{H_l \times W_l \times C_l}$ as the output features of the $l$-th layer, obtained from $\Phi(x_{n,i})$, where $i$ denotes the $i$-th sample from local dataset, $l$ denotes the $l$-th layer and $H_l$, $W_l$, $C_l$ are the size and channel of the feature map. In multi-level feature fusion, we employ a straightforward approach combining up-sampling and feature concatenation to extract features from pre-trained CNNs. This process can be  expressed as follows:
\begin{equation}
\label{eq8}
\tilde{\mathbf{F}}_{n,i} = \operatorname{Concat}(\{ \operatorname{Up}({\mathbf{F}}_{n,i,l},(H_1,W_1))\} _{l = 1}^L),
\end{equation}
where $\operatorname{Concat}(\cdot)$ represents the concatenation operation, $\operatorname{Up}(\cdot,(H,W))$ denotes bilinear interpolation to the spatial dimensions $(H, W)$ and $(H_1,W_1)$ is the feature resolution of $l=1$.

\noindent\textbf{Projection Layer.} Feature extractors pretrained on large-scale natural image datasets often exhibit significant distribution shifts when applied to industrial or medical images. To address this issue, we introduce a projection layer $\Psi(\cdot;\theta_{p,n})$ designed to adapt the extracted features to the anomaly detection domain. As shown in Fig.~\ref{fig_3}, the projection layer is implemented as a $1 \times 1$ convolutional layer followed by an activation function. The output of $\Psi(\tilde{\mathbf{F}}_{n,i};\theta_{p,n})$ is denoted as $\mathbf{Z}_{n,i}$.

\noindent\textbf{Memory Generator.} The memory bank-based approach (\cite{cohen2020sub, defard2021padim, roth2022towards}) has demonstrated significant performance in UAD. However, the memory banks utilized in these existing methods are non-trainable, meaning they remain unchanged following initialization. Due to the local bias in feature distributions for federated UAD, using the static memory bank will lead to overfitting on domain-specific data, increasing the risk of error. Furthermore, existing methods operate mainly on a discrete feature space, which increases the variations in feature distributions between clients for normal samples (\cite{lee2025continuous}). As shown in Fig.~\ref{fig_3}, we propose a memory generator that integrates spatial information and ensures the continuity of feature space. The memory generator $\Gamma(\cdot;{\theta_{g,n}, \mathbf G_{n}})$ utilizes a coordinate convolution layer to encode spatial location information by incorporating additional coordinate channels, thereby facilitating the network's ability to learn spatial transformations more effectively (\cite{liu2018intriguing}). The output of coordinate convolution is represented as $\hat{\mathbf{Z}}_{n,i}$ using the parameter $\theta_{g,n}$.

Then, to enhance the continuity of the memory bank, the memory generator utilizes a grid-based approach to construct a continuous feature space. We define a trainable grid space as $\mathbf G_{n} \in {\mathbb{R}^{H_{g}\times W_{g}\times C}}$, where $H_{g}$ and $W_{g}$ are hyperparameters indicating the size of the continuous space. The memory generator maps $\hat{\mathbf{Z}}_{n,i}$ to pixel-wise coordinates $\mathbf{U}_{n,i}\in {\mathbb{R}^{H_1 \times W_1 \times 2}}$ through two $1\times1$ convolutional layers in our memory generator. A sample function ${s}(\cdot,\cdot)$ is employed to extract features from the continuous space ${\mathbf G}_n$ at specific coordinate values. For each spatial location $(h,w)$ in $\mathbf{U}_{n,i}$, we denote the predicted coordinate as $(u_x,u_y)=\mathbf{U}_{n,i}[h,w]$. The normalized coordinates are computed as follows:
\begin{equation}
\label{eq9}
\tilde{u}_x = \frac{(u_x + 1)}{2} \cdot (W_g - 1), \quad
\tilde{u}_y = \frac{(u_y + 1)}{2} \cdot (H_g - 1),
\end{equation}
and the result of $s(\mathbf{G}_n, \mathbf{U}_{n,i})$ is defined as ${\mathbf O}_{n,i}$. Each pixel value in ${\mathbf O}_{n,i}$ can be sampled as:
\begin{equation}
\label{eq10}
\mathbf{O}_{n,i,h,w} = \sum\limits_{a = 1}^2 \sum\limits_{b = 1}^2
\omega_{ab}{\mathbf G_n}_{\left\lfloor \tilde{u}_y \right\rfloor + (a-1),\left\lfloor \tilde{u}_x \right\rfloor + (b-1)},
\end{equation}
where
\begin{equation}
\begin{aligned}
\omega_{ab} = &(1 - |\tilde{u}_x - ( \left\lfloor \tilde{u}_x \right\rfloor + (b-1))|) \\
&(1 - |\tilde{u}_y - (\left\lfloor \tilde{u}_y \right\rfloor + (a-1))|), \notag
\end{aligned}
\end{equation}
and $\left\lfloor {\cdot} \right\rfloor$ represents the floor operation. Finally, the features ${\mathbf O}_{n,i}$, sampled from the continuous feature space, are concatenated with the coordinate convolution outputs $\hat{\mathbf{Z}}_{n,i}$. This concatenated tensor is then passed through a convolutional layer to generate the memory feature ${\mathbf M}_{n,i} \in \mathbb{R}^{H_1 \times W_1 \times C}$ for client $n$.

\subsection{Memory-Reduce before Uploading}
\label{memory-reduce}
The earlier memory-based methods (\cite{cohen2020sub, roth2022towards}) involved storing all sample memory features from the training dataset within the memory bank. Accordingly, the capacity of the memory bank scales proportionally with the size of the local dataset. For federated UAD, this design leads to high communication overhead when applied to large-scale datasets and also introduces privacy risks. When each memory feature corresponds directly to a raw input sample, it becomes vulnerable to privacy leakage under reconstruction or inversion attacks. Moreover, the varying sizes of memory banks also introduce challenges in the aggregation process. Motivated by this, FedDyMem incorporates a memory-reduce method to compress the collected memory features from the local dataset, thereby reducing the mutual information between the memory features and the raw data.

Motivated by this, after local training at the $t$-th communication round, FedDyMem applies memory-reduce on each client to compress the collected  memory features $\{{\mathbf{M}_{n,i}^{(t)}}\}_{i=1}^{\left| {\mathcal D}_n^{\mathrm{train}} \right|} \in {\mathbb{R}^ {{\left|{\cal D}_n^{\mathrm{train}} \right|} \times H_1 \times W_1 \times C}} $ of all samples into a fixed-size reduced local memory bank ${\mathbf{R}_n^{(t)}}$. Only the reduced memory bank is uploaded to the server for subsequent aggregation.

\begin{figure}[t]
\centering
\includegraphics[width=0.45\textwidth]{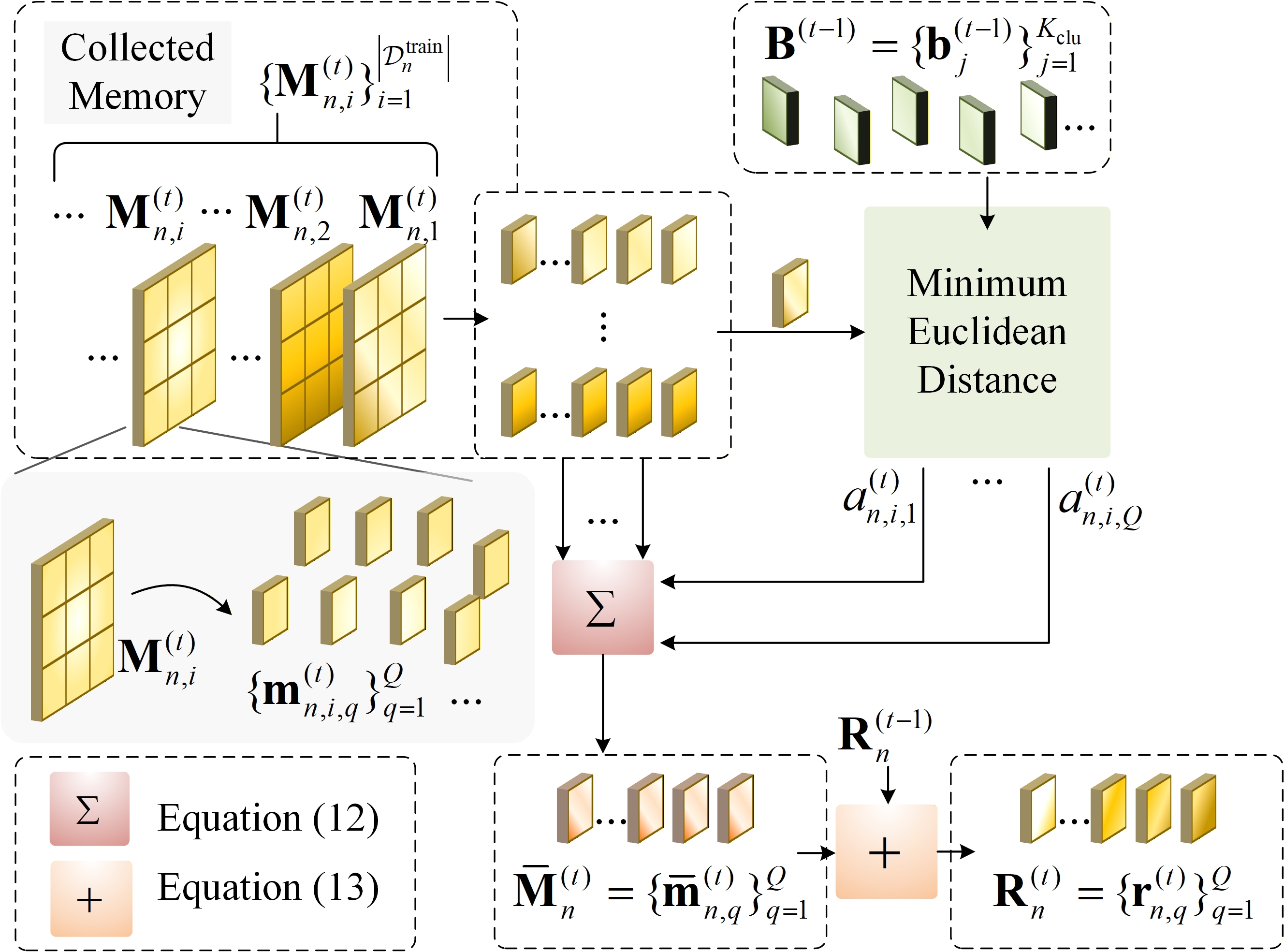} 
\caption{Illustration of Memory-reduce.}
\label{fig_4}
\end{figure}

To address the memory bank size issue, (\cite{lee2022cfa}) employs an Exponential Moving Average (EMA) method. However, the reduced memory bank remains sensitive to the sequence in which input samples are processed. Thus, we introduce a dynamic weighted average that optimizes memory size while preserving sample order independence and effectively correlating with the number of communication rounds. As shown in Fig.~\ref{fig_4}, we apply the dynamic weighted averaging operation to compute the $\bar{\mathbf{M}}_{n}^{(t)}$. Let $Q=H_1W_1$ denote the number of patch features in each local memory feature and ${{\mathbf{M}}}_{n,i}^{(t)} = \{\mathbf{m}_{n,i,q}^{(t)}\}_{q=1}^Q$. The global memory contains $K_{\mathrm{clu}}$ prototypes, denoted as $\mathbf{B}^{(t-1)} = \{\mathbf{b}_j^{(t-1)}\}_{j=1}^{K_{\mathrm{clu}}}$ (detailed in Section~\ref{server}). The dynamic weight for $q$-th patch feature is computed by the nearest distance:
\begin{equation}
\label{eq11}
a_{n,i,q}^{(t)}
=
\min\limits_{1\le j\le K_{\mathrm{clu}}}
\operatorname{dis}
\left(
\mathbf m_{n,i,q}^{(t)},\mathbf b_j^{(t-1)}
\right),
\end{equation}
where the $\operatorname{dis}(\cdot,\cdot)$ denotes the Euclidean distance. The average memory is denoted as $\bar{\mathbf{M}}_{n}^{(t)} = \{\bar{\mathbf{m}}_{n,q}^{(t)}\}_{q=1}^{Q}$, where each average patch feature is computed by: 
\begin{equation}
\label{eq12}
    \bar{\mathbf{m}}_{n,q}^{(t)} = \sum\limits_{i = 1}^{\left| {{{\cal D}}_n^{\mathrm{train}}} \right|}{{a_{n,i,q}^{(t)}}{{{\mathbf{m}}}_{n,i,q}^{(t)}}} /\sum\limits_{i = 1}^{\left| {{{\cal D}}_n^{\mathrm{train}}} \right|} a_{n,i,q}^{(t)},
\end{equation}

Thus, the average feature $\bar{\mathbf{M}}_{n}^{(t)}$ is influenced by the distance between the memory representation $\mathbf{M}_{n}^{(t)}$ and the previous memory bank $\mathbf{B}^{(t-1)}$. 

After computing $\bar{\mathbf{M}}_{n}^{(t)}$, a round-based EMA (rEMA) method is employed to update the memory bank. Each client keeps a local cache of its previously uploaded reduced memory bank $\mathbf{R}_n^{(t-1)}$, which is used as the historical term in the rEMA update.
This approach ensures smoother and more stable updates during the training process. The update rule is defined as:
\begin{equation}
\label{eq13}
    \mathbf{R}_n^{(t)} = \alpha_t{\bar{\mathbf{M}}_{n}^{(t)}} + (1-\alpha_t){{{{\mathbf{R}}}_n^{(t-1)}}},
\end{equation}
where $\alpha_t$ represents the exponential decay rate, dynamically calculated as  $\alpha_t = 1/ (t+1)$. 

\subsection{Memory Aggregation on Server}
\label{server}
At round $t$, the server $\cal S$ collects the reduced memory banks from $N$ clients, denoted as ${\{\mathbf{R}_n^{(t)}\}}_{n=1}^N$. The objective of the server is to aggregate these memory banks into a unified, representative memory bank $\mathbf{B}^{(t)}$. Representation-based federated learning approaches, such as those proposed in (\cite{tan2022fedproto, tan2022federated}), utilize weighted aggregation of class prototypes during model updates. However, significant feature bias in the early training stages leads to aggregated memory that deviates significantly from the client-specific memory, which often contains incomplete product types. This misalignment hinders convergence, posing a challenge for clients with limited data diversity. Thus, we employ a clustering-based approach to aggregate the collected memory banks effectively. 

Let $Q=H_1W_1$ denote the number of local reduced prototypes in each uploaded memory bank. Therefore, the server clusters only $NQ$ uploaded prototypes. Each reduced memory is flattened into a set of $Q$ patch-level prototypes, denoted as $\mathbf{R}_n^{(t)} = \{\mathbf{r}_{n,q}^{(t)}\in\mathbb R^C\}_{q=1}^{Q}$. The server then collects all uploaded prototypes as:
\begin{equation}
\label{eq14}
    \mathcal{P}^{(t)} = 
    \bigcup\nolimits_{n=1}^{N}
    \{\mathbf{r}_{n,q}^{(t)}\}_{q=1}^{Q}.
\end{equation}
FedDyMem applies k-means clustering to $\mathcal{P}^{(t)}$ and uses the cluster centers as global memory prototypes:
\begin{equation}
\label{eq15}
\mathbf{B}^{(t)}
=
\{\mathbf{b}_j^{(t)} \in\mathbb{R}^C\}_{j=1}^{K_{\mathrm{clu}}}
=
\operatorname{Kmeans}(\mathcal{P}^{(t)}, K_{\mathrm{clu}}),
\end{equation}
where $K_{\mathrm{clu}}$ denotes the number of global prototypes after aggregation. The clustering results describe the overall distribution and similarity relationships among patch-level prototypes stored across all reduced local memory banks. By maintaining multiple representative centers, the aggregated memory bank can cover different normal variations across clients more effectively than a single averaged representation.

Since clustering is performed on the reduced memory banks, the server processes $NQ$ uploaded prototypes in each round. The computational complexity of k-means aggregation is $\mathcal O(I_{\mathrm{clu}}NQK_{\mathrm{clu}}C)$, where $I_{\mathrm{clu}}$ is the number of k-means iterations.

\begin{algorithm}[thp]
\begingroup
\caption{Initialization of FedDyMem at Round $t=0$.}
\label{alg_init}
\begin{algorithmic}[1]
\renewcommand{\algorithmicrequire}{\textbf{Input:}}
\renewcommand{\algorithmicensure}{\textbf{Output:}}
\Require Training sets $\{\mathcal D_n^{\mathrm{train}}\}_{n=1}^{N}$, server $\mathcal S$, number of global prototypes $K_{\mathrm{clu}}$.
\Ensure Initial global memory bank $\mathbf B^{(0)}$, initial reduced local memories $\{\mathbf R_n^{(0)}\}_{n=1}^{N}$, and initialized local modules $\{\theta_{p,n}^{(0)},\theta_{g,n}^{(0)},\mathbf G_n^{(0)}\}_{n=1}^{N}$.
\State Initialize $(\theta_{p,n}^{(0)}, \theta_{g,n}^{(0)}, \mathbf G_n^{(0)})$ for each client with the same initialization $(\theta_p^{(0)},\theta_g^{(0)},\mathbf G^{(0)})$.
\For{each client $n=1,\ldots,N$ \textbf{in parallel}}
    \For{$i=1$ to $\left|\mathcal{D}_n^{\mathrm{train}}\right|$}
        \State $\{\mathbf F_{n,i,l}^{(0)}\}_{l=1}^{L} \gets \Phi(x_{n,i})$;
        \State $\tilde{\mathbf F}_{n,i}^{(0)} \gets$ Compute by Equation~\eqref{eq8};
        \State $\mathbf M_{n,i}^{(0)} \gets \Gamma\!\left(\Psi\!\left(\tilde{\mathbf F}_{n,i}^{(0)};\theta_{p,n}^{(0)}\right);\theta_{g,n}^{(0)},\mathbf G_n^{(0)}\right)$;
    \EndFor
    \State $\mathbf R_n^{(0)} \gets$ Compute the initial reduced local memory bank by Equation~\eqref{init_mem_reduced};
    \State Upload $\mathbf R_n^{(0)}$ to the server $\mathcal S$;
\EndFor
    \State Construct $\mathcal P^{(0)}=\bigcup_{n=1}^{N}\{\mathbf r_{n,q}^{(0)}\}_{q=1}^{Q}$ by flattening $\{\mathbf R_n^{(0)}\}_{n=1}^{N}$ on $\mathcal S$;
    \State $\mathbf B^{(0)} \gets \operatorname{Kmeans}\!\left(\mathcal P^{(0)},K_{\mathrm{clu}}\right)$; \\
    \Return $\mathbf B^{(0)}, \{\mathbf R_n^{(0)}\}_{n=1}^{N}, \{\theta_{p,n}^{(0)},\theta_{g,n}^{(0)},\mathbf G_n^{(0)}\}_{n=1}^{N}$.
\end{algorithmic}
\endgroup
\end{algorithm}

\subsection{Training Pipeline}
\noindent\textbf{Initialization.} The initialization stage aims to construct the first shared reference memory bank $\mathbf{B}^{(0)}$ before any iterative local optimization is performed. We denote this stage as round $t=0$. Different from the subsequent training rounds, no metric loss is computed at $t=0$, and the trainable projection layer and memory generator are not updated in this stage.

To ensure that the initial memory features are represented in a consistent feature space, all clients start from the same initialization. Specifically, each client employs the same frozen backbone and starts with identical projection parameters $\theta_p^{(0)}$, memory generator parameters $\theta_g^{(0)}$, and grid space $\mathbf G^{(0)}$. For each normal training sample $x_{n,i}\in\mathcal D_n^{\mathrm{train}}$, client $n$ first extracts multi-level features using the frozen pretrained backbone $\Phi(\cdot)$ and obtains the fused feature $\tilde{\mathbf F}_{n,i}^{(0)}$. The initial per-sample memory feature $\mathbf M_{n,i}^{(0)}$ is then generated by the initialized projection layer and memory generator, denoted as $\mathbf M_{n,i}^{(0)}=\Gamma(
\Psi(\tilde{\mathbf F}_{n,i}^{(0)};\theta_{p,n}^{(0)});
\theta_{g,n}^{(0)},\mathbf G_n^{(0)})$.

After obtaining memory features, each client applies memory-reduce to construct local memory bank $\mathbf R_n^{(0)}$. Since no previous global memory bank exists at $t=0$, the dynamic weights are set as $a_{n,i,q}^{(0)}=1$, and the rEMA update is not applied. Thus, the reduced local memory bank is computed as:
\begin{equation}
\label{init_mem_reduced}
\mathbf R_n^{(0)}
=
\frac{1}{\left|{{{\cal D}}_n^{\mathrm{train}}}\right|}
\sum\limits_{i = 1}^{\left| {{{\cal D}}_n^{\mathrm{train}}} \right|}
\mathbf M_{n,i}^{(0)}.
\end{equation}
The server collects the reduced memory banks from all clients and constructs the initial global memory bank $\mathbf{B}^{(0)}$ through the aggregation described in Section~\ref{server}. Because all clients share the same initial projection, memory-generator, and grid parameters, the initial memory features are not arbitrary random targets. Conditioned on this shared initialization, they are deterministic, data-dependent transformations of the frozen pretrained backbone features in a common feature coordinate system. After memory-reduce and server-side $K$-means aggregation, the initial global memory bank $\mathbf{B}^{(0)}$ summarizes the aggregated initial representations of normal training data across clients. It is then distributed as the fixed reference memory for the first trainable round ($t=1$). The initialization is summarized in Algorithm~\ref{alg_init}.
 
\noindent\textbf{Local Training.}
After initialization, FedDyMem performs iterative local training. For each trainable communication round $t\ge1$, client $n$ receives the global memory bank $\mathbf B^{(t-1)}$ from the server. Given a normal sample $x_{n,i}\in\mathcal D_n^{\mathrm{train}}$, client $n$ generates the current memory feature $\mathbf{M}_{n,i}^{(t)}$ by the backbone, the projection layer and the memory generator. During this process, the pretrained backbone $\Phi(\cdot)$ is frozen. The trainable components are the projection layer parameters $\theta_{p,n}^{(t)}$, the memory generator parameters $\theta_{g,n}^{(t)}$, and the trainable grid space $\mathbf G_n^{(t)}$.

Considering the bias in feature distributions across clients, we introduce a metric loss to optimize the local model to produce normal features that are more closely aligned with the shared global memory bank. For each local memory feature, we use $Q = H_1W_1$ as the number of patch features. Thus, the memory feature $\mathbf{M}_{n,i}^{(t)}$ can be denoted as $\{\mathbf{m}_{n,i,q}^{(t)}\}_{q=1}^{Q}$ and the reference global memory $\mathbf{B}^{(t-1)}$ can be represented as $\{\mathbf{b}_{j}^{(t-1)}\}_{j=1}^{K_{\mathrm{clu}}}$. 
For each local patch feature $\mathbf{m}_{n,i,q}^{(t)}$, its $k_{\mathrm{nn}}$ nearest neighbors are retrieved from the $K_{\mathrm{clu}}$ global prototypes in $\mathbf{B}^{(t-1)}$. We denote its $r$-th nearest neighbor as $\mathbf b_{n,i,q,r}^{(t-1)}$ and the metric loss is then defined as:
\begin{equation}
\label{loss}
\begin{aligned}
&\ell(\mathbf M_{n,i}^{(t)},\mathbf B^{(t-1)}) \\
&=\frac{1}{Qk_{\mathrm{nn}}}\sum\limits_{q=1}^{Q}\sum\limits_{r=1}^{k_{\mathrm{nn}}}
\max(0,\operatorname{dis}(\mathbf m_{n,i,q}^{(t)},\mathbf b_{n,i,q,r}^{(t-1)})
-
\tau_\ell),
\end{aligned}
\end{equation}
where $k_{\mathrm{nn}}$ is the number of nearest neighbors, and $\tau_\ell$ is the margin threshold. The local objective of client $n$ at round $t$ is then defined as:
\begin{equation}
\label{total-loss}
\mathcal L_n^{(t)}=\frac{1}{\left|{{{\cal D}}_n^{\mathrm{train}}}\right|}\sum_{i=1}^{\left|{{{\cal D}}_n^{\mathrm{train}}}\right|}
\ell(
\mathbf M_{n,i}^{(t)},\mathbf B^{(t-1)}
).
\end{equation}

\begin{algorithm}[thp]
\begingroup
\caption{FedDyMem Training and Memory Communication.}
\label{alg_train}
\begin{algorithmic}[1]
\renewcommand{\algorithmicrequire}{\textbf{Input:}}
\renewcommand{\algorithmicensure}{\textbf{Output:}}
\Require Training sets $\{\mathcal D_n^{\mathrm{train}}\}_{n=1}^{N}$, server $\mathcal S$, communication rounds $T$, local epochs $E$, number of nearest neighbors $k_{\mathrm{nn}}$; loss threshold $\tau_\ell$; number of global prototypes $K_{\mathrm{clu}}$.
\Ensure Final global memory bank $\mathbf B^{(T)}$ and trained local modules $\{\theta_{p,n}^{(T)},\theta_{g,n}^{(T)},\mathbf G_n^{(T)}\}_{n=1}^{N}$.
\State $\mathbf B^{(0)},\{\mathbf R_n^{(0)}\}_{n=1}^{N}, \{\theta_{p,n}^{(0)},\theta_{g,n}^{(0)},\mathbf G_n^{(0)}\}_{n=1}^{N } \gets$ Obtain by Algorithm~\ref{alg_init};
\For{$t=1$ to $T$}
    \For{each client \textbf{in parallel}}
        \State $(\theta_{p,n}^{(t)},\theta_{g,n}^{(t)},\mathbf G_n^{(t)})
        \gets$ \textsc{Update}(
        \Statex \hspace*{\algorithmicindent}$n,t,E,\mathcal D_n^{\mathrm{train}},\mathbf B^{(t-1)}$);
        \State $\mathbf R_n^{(t)}\gets$ \textsc{Reduce}(
        \Statex \hspace*{\algorithmicindent}$n,t,\mathcal D_n^{\mathrm{train}},\mathbf B^{(t-1)},\mathbf R_n^{(t-1)}, \theta_{p,n}^{(t)},\theta_{g,n}^{(t)},\mathbf G_n^{(t)}$);
        \State Upload $\mathbf R_n^{(t)}$ to the server $\mathcal S$;
    \EndFor
    \State Server $\mathcal S$ collects uploaded reduced memory banks $\{\mathbf R_n^{(t)}\}_{n=1}^{N}$;
    \State Construct $\mathcal P^{(t)}$ by Equation~\eqref{eq14};
    \State $\mathbf B^{(t)} \gets \operatorname{Kmeans}\!\left(\mathcal P^{(t)},K_{\mathrm{clu}}\right)$;
    \State Distribute $\mathbf B^{(t)}$ to all participating clients for the next round;
\EndFor
\State Distribute the final global memory bank $\mathbf B^{(T)}$ to all clients. \\
\Return $\mathbf B^{(T)}, \{\theta_{p,n}^{(T)},\theta_{g,n}^{(T)},\mathbf G_n^{(T)}\}_{n=1}^{N}$
\renewcommand{\algorithmicrequire}{\underline{\textsc{Update}$(n,t,E,\mathcal D_n^{\mathrm{train}},\mathbf B^{(t-1)})$}}
\Require
\For{$e=1$ to $E$}
    \For{$i=1$ to $\left|\mathcal{D}_n^{\mathrm{train}}\right|$}
        \State $\mathbf{M}_{n,i}^{(t)} \gets$ Generate memory feature by $\Phi(\cdot), \Psi(\cdot;\theta_{p,n}^{(t)}), \Gamma(\cdot;\theta_{g,n}^{(t)},\mathbf G_n^{(t)})$;
        \State $\ell(\mathbf M_{n,i}^{(t)},\mathbf B^{(t-1)}) \gets$ Compute the local metric loss by Equation~\eqref{loss};
        \State Update $\theta_{p,n}^{(t)}$, $\theta_{g,n}^{(t)}$, and $\mathbf G_n^{(t)}$ by back-propagation, while keeping $\Phi(\cdot)$ frozen;
    \EndFor
\EndFor \\
\Return $\theta_{p,n}^{(t)}, \theta_{g,n}^{(t)}, \mathbf G_n^{(t)}$
\renewcommand{\algorithmicrequire}{\underline{\textsc{Reduce}$(n,t,\mathcal D_n^{\mathrm{train}},\mathbf B^{(t-1)},\mathbf R_n^{(t-1)},\theta_{p,n}^{(t)}, \theta_{g,n}^{(t)}, \mathbf G_n^{(t)})$}}
\Require
\For{$i=1$ to $\left|\mathcal{D}_n^{\mathrm{train}}\right|$}
            \State $\mathbf M_{n,i}^{(t)} \gets$ Regenerate using the frozen $\Phi(\cdot)$ and the updated 
            $\Psi(\cdot;\theta_{p,n}^{(t)})$, 
            $\Gamma(\cdot;\theta_{g,n}^{(t)},\mathbf G_n^{(t)})$;
        \EndFor
\State $\alpha_t \gets 1/(t+1)$;
\State $\mathbf R_n^{(t)} \gets$ Compute the reduced memory by Equations~\eqref{eq11}--\eqref{eq13} with $\{\mathbf M_{n,i}^{(t)}\}_{i=1}^{\left|\mathcal{D}_n^{\mathrm{train}}\right|},\mathbf B^{(t-1)},\mathbf R_n^{(t-1)},\alpha_t$; \\
\Return $\mathbf R_n^{(t)}$
\end{algorithmic}
\endgroup
\end{algorithm}

\noindent\textbf{Communication.}  
After $E$ local training epochs at communication round $t\ge1$, FedDyMem performs memory-based communication between clients and the server. The objective of this step is to update the shared global memory bank without exchanging local model parameters or gradients.

For client $n$, after local optimization is completed, the updated projection layer and memory generator are kept fixed and used to regenerate the memory features of all normal training samples, denoted as $\{\mathbf M_{n,i}^{(t)}\}_{i=1}^{\left|{{{\cal D}}_n^{\mathrm{train}}}\right|}$. These memory features are not directly uploaded. Instead, client $n$ applies the memory-reduce defined in Section~\ref{memory-reduce} to obtain the local memory bank $\mathbf{R}_n^{(t)}$.

After receiving the uploaded reduced memory banks, the server aggregates them into a new global memory bank $\mathbf{B}^{(t)}$ using the aggregation described in Section~\ref{server}. The updated global memory bank $\mathbf B^{(t)}$ is then distributed to all clients and used as the fixed reference memory for the next round of local training. The complete training and communication protocol is summarized in Algorithm~\ref{alg_train}.

\subsection{Testing and Anomaly Scoring}
After the final communication round $T$, FedDyMem performs inference using local model trained on each client together with a shared global memory bank, rather than using a single global anomaly detector obtained by parameter aggregation. As shown in Fig.~\ref{fig_2}(b), each client receives the globally aggregated memory bank $\mathbf{B}^{(T)}$ from the server at the end of training. Following prior works (\cite{roth2022towards, lee2022cfa}), we utilize a nearest neighbor search in the feature space to calculate the anomaly score for each test sample in ${{\cal D}^{\mathrm{test}}}$.

Given a test image $x^{\mathrm{test}}$ on $n$-th client, the well trained projection layer $\Psi(\cdot;\theta_{p,n}^{(T)})$ and memory generator $\Gamma(\cdot;\theta_{g,n}^{(T)},\mathbf G_n^{(T)})$ are used to extract the memory feature representation $\mathbf{M}^{\mathrm{test}}$ following the same procedure as in training. For each spatial location, the corresponding feature patch $\mathbf{m}^{\mathrm{test}}_{q} \in \mathbb{R}^C$ is compared with the final global memory bank $\mathbf{B}^{(T)}$ to compute its anomaly score. The global memory bank is represented as an unordered set of normal prototypes $\{\mathbf{b}_{j}^{(T)} \in \mathbb{R}^C\}_{j=1}^{K_{\mathrm{clu}}}$. Following previous works, the patch-level anomaly score $s_q$ is computed as the minimum Euclidean distance to the nearest global prototype $\min_{1\le j\le  K_{\mathrm{clu}}}\operatorname{dis}(\mathbf{m}^{\mathrm{test}}_{q},\mathbf b_j^{(T)})$. For anomaly localization, the patch-level scores are upsampled to the original image resolution to obtain the anomaly map. The image-level anomaly score is computed by the softmax-weighted maximum strategy:
\begin{equation}
\label{ann_score}
s=
\max_{1\le q\le Q}\left(s_q\cdot\frac{\exp(s_q)}{\sum_{q'=1}^{Q}\exp(s_{q'})}\right).
\end{equation}

\subsection{Privacy Analysis}
 In federated learning, Mutual Information(MI) is a fundamental tool for evaluating privacy. Prior studies (\cite{zhang2023fedcr, alsulaimawi2024federated,tan2024defending}) indicate that MI provides a principled paradigm for privacy analysis and can jointly improve privacy and utility. Consistent with prior work, we adopt an MI–based methodology to analyze FedDyMem’s privacy. For two statistics $X$ and $Y$, their mutual information $I(X;Y)$ quantifies statistical dependence; under Gaussian assumptions,
\begin{equation}
    I(X;Y) = -\tfrac{1}{2}\log\!\big(1-\rho_{xy}^2\big).
\end{equation}
where $\rho_{xy}$ denotes the Pearson correlation coefficient. A smaller $I(X;Y)$ implies that an adversary can extract less information about $X$ from observing $Y$, indicating stronger privacy.

To make the following analysis well-defined for tensor-valued images and memory banks, we consider one fixed scalar statistic at a time. Specifically, for a fixed client $n$, round $t$, patch $q$, and channel $c$, let $X_j=\phi(x_{n,j})$ be a scalar statistic of the $j$-th raw sample and let $Y_i=[\mathbf m_{n,i,q}^{(t)}]_c$ be the corresponding scalar memory statistic. The result below is therefore a conditional, per-coordinate dependence analysis.

Let $D_n=|\mathcal D_n^{\mathrm{train}}|$ and, for a non-degenerate coordinate with $\sum_{i=1}^{D_n}a_{n,i,q}^{(t)}>0$, define $\tilde w_i=a_{n,i,q}^{(t)}/\sum_{r=1}^{D_n}a_{n,r,q}^{(t)}$ and $\bar Y=\sum_{i=1}^{D_n}\tilde w_iY_i$. Let $\mathcal H^{(t)}$ denote the information fixed for the current reduction, including the realized normalized weights, the previous reduced memory, the previous global memory, and the local encoder state. By Eq.~\eqref{eq13}, the transmitted coordinate is $Y_{\mathrm{red}}=[\mathbf R_n^{(t)}]_{q,c}=\alpha_t\bar Y+(1-\alpha_t)Y_{\mathrm{hist}}$, where $Y_{\mathrm{hist}}=[\mathbf R_n^{(t-1)}]_{q,c}$ is fixed conditional on $\mathcal H^{(t)}$; at $t=0$, $Y_{\mathrm{red}}=\bar Y$. Conditional on $\mathcal H^{(t)}$, memory-reduce is a deterministic post-processing of $\{Y_i\}_{i=1}^{D_n}$, and DPI gives
\begin{equation}
I\!\left(X_j;Y_{\mathrm{red}}\mid\mathcal H^{(t)}\right)
\le
I\!\left(X_j;\{Y_i\}_{i=1}^{D_n}\mid\mathcal H^{(t)}\right).
\end{equation}
The following lemma gives a quantitative comparison with the corresponding per-sample scalar memory statistic under the stated distributional assumptions.

\begin{lemma}[Correlation Quantitative Bound]\label{MI-privacy-lemma2}
Let $\{Y_i\}_{i=1}^N$ be scalar statistics and define the normalized weighted average $Y=\sum_{i=1}^N \tilde w_i Y_i$ with $\tilde w_i\ge0$, $\sum_i\tilde w_i=1$. All moments, correlations, and independence statements below are conditional on $\mathcal H^{(t)}$; this conditioning is omitted in the notation. The normalized weights are thus treated as fixed in the bound.
Assume:
\begin{enumerate}
    \item $X_j$ is independent of $Y_i$ for $i\neq j$.
    \item Each $Y_i$ has finite, strictly positive variance $\sigma_{Y_i}^2$, and $\mathrm{Cov}(Y_i,Y_k)\ge0$ for $i\neq k$.
\end{enumerate}
Denote $\sigma_{Y_i}^2=\mathrm{Var}(Y_i)$, then
\begin{equation}
|\rho(X_j, Y)| \;\le\;  
\frac{\tilde w_j\,\sigma_{Y_j}}{\sqrt{\sum_{i=1}^N \tilde w_i^2\,\sigma_{Y_i}^2}}
\,|\rho(X_j, Y_j)|.
\end{equation}
\end{lemma}

\begin{proof}
By linearity and independence, $\mathrm{Cov}(X_j,Y)=\tilde w_j\,\mathrm{Cov}(X_j,Y_j)$, hence 
\begin{equation}
\label{eq23}
|\rho(X_j,Y)| 
= \frac{|\mathrm{Cov}(X_j,Y)|}{\sigma_{X_j}\,\sigma_Y}
= \frac{\tilde w_j\,|\mathrm{Cov}(X_j,Y_j)|}{\sigma_{X_j}\,\sigma_Y}.
\end{equation}

For the denominator,
\begin{equation}
\mathrm{Var}(Y) = \sum_{i=1}^N \tilde w_i^2\,\sigma_{Y_i}^2 
+ 2\sum_{1\le i<k\le N} \tilde w_i \tilde w_k\,\mathrm{Cov}(Y_i,Y_k).
\end{equation}
Since $\mathrm{Cov}(Y_i,Y_k)\ge 0$, with $\sigma_Y = \sqrt{\mathrm{Var}(Y)}$, the variance satisfies
$\sigma_Y = \sqrt{\mathrm{Var}(Y)} \ge \sqrt{\sum_{i=1}^N \tilde w_i^2\,\sigma_{Y_i}^2}$.

Substituting this bound into the denominator of \eqref{eq23} with $|\mathrm{Cov}(X_j,Y_j)|=\sigma_{X_j}\sigma_{Y_j}|\rho(X_j,Y_j)|$, completes the proof.
\end{proof}

\begin{theorem}[Memory-Reduce Privacy]\label{thm:privacy}
Under the conditional Gaussian approximation used above and the conditions of Lemma~\ref{MI-privacy-lemma2}, the reduced scalar statistic satisfies
\begin{equation}
I\!\left(X_j;Y_{\mathrm{red}}\mid\mathcal H^{(t)}\right)
\le
I\!\left(X_j;Y_j\mid\mathcal H^{(t)}\right).
\end{equation}
The inequality is strict whenever $0<|\rho(X_j,Y_j\mid\mathcal H^{(t)})|<1$ and the target sample does not provide the only nonzero effective contribution to the weighted average.
\end{theorem}

\begin{proof}For the fixed $(n,t,q,c)$, Eq.~\eqref{eq13} gives $Y_{\mathrm{red}}=\alpha_t\bar Y+(1-\alpha_t)Y_{\mathrm{hist}}$. Since $Y_{\mathrm{hist}}$ is fixed conditional on $\mathcal H^{(t)}$ and $\alpha_t>0$, $|\rho(X_j,Y_{\mathrm{red}}\mid\mathcal H^{(t)})|=|\rho(X_j,\bar Y\mid\mathcal H^{(t)})|$. Applying Lemma~\ref{MI-privacy-lemma2} to $\bar Y$ yields
\begin{equation}
\label{eq28}
\left|\rho\!\left(X_j,Y_{\mathrm{red}}\mid\mathcal H^{(t)}\right)\right|
\le
\gamma_j\left|\rho\!\left(X_j,Y_j\mid\mathcal H^{(t)}\right)\right|,
\end{equation}
where
\begin{equation}
\gamma_j^2=
\frac{\tilde w_j^2\sigma_{Y_j}^2}{\sum_{i=1}^{D_n}\tilde w_i^2\sigma_{Y_i}^2}
\le 1. 
\end{equation}

Here, $\gamma_j<1$ whenever at least one other sample has nonzero normalized weight and positive variance. When $0<|\rho(X_j,Y_j\mid\mathcal H^{(t)})|<1$, the conditional Gaussian MI formula is monotone in $|\rho|$, and therefore
\begin{equation}
\frac{I\!\left(X_j;Y_{\mathrm{red}}\mid\mathcal H^{(t)}\right)}
 {I\!\left(X_j;Y_j\mid\mathcal H^{(t)}\right)}
\le
\frac{-\tfrac{1}{2}\log\!\left(1-\gamma_j^2\rho\!\left(X_j,Y_j\mid\mathcal H^{(t)}\right)^2\right)}
 {-\tfrac{1}{2}\log\!\left(1-\rho\!\left(X_j,Y_j\mid\mathcal H^{(t)}\right)^2\right)}.
\end{equation}
The right-hand side is no greater than $1$ because $0\leq\gamma_j\leq1$, and it is strictly smaller than $1$ under the condition above. The zero-correlation case is immediate. This proves the stated conditional bound.
\end{proof}

For an asymptotic statement, an additional non-dominance condition is required: $\gamma_j^2\to0$ as the number of contributing samples grows. This condition is satisfied, for example, when normalized weights remain approximately balanced and the conditional variances are bounded above and away from zero. Under this condition and the assumptions above, $|\rho(X_j,Y_{\mathrm{red}}\mid\mathcal H^{(t)})|\to0$ and $I(X_j;Y_{\mathrm{red}}\mid\mathcal H^{(t)})\to0$. Thus, memory-reduce attenuates the conditional per-sample dependence whenever no single sample dominates the weighted aggregate. Furthermore, FedDyMem exchanges statistical summaries (memory banks) rather than model parameters (\cite{tan2022fedproto}), while the reconstruction, membership, and client-inference attacks in Section~\ref{sec:threat_model} provide complementary empirical evidence of reduced leakage.

\section{Convergence Analysis}
\label{convergence}
We analyze the convergence of FedDyMem under a non-convex loss.
Let the loss at local iteration $e$ in communication round $t$ be
$\mathcal{L}_{e}^{(t)} := \mathbb E_{x\sim\mathcal D_n}\!\big[\ell(\mathbf{M}_e^{(t)},\ \mathbf{B}^{(t-1)})\big]$,
where $\mathbf{M}_{e}^{(t)}=f(x;\theta_{e}^{(t)})$ and its Jacobian as $J_f(x;\theta_{e}^{(t)})$.
Here $\mathbf{B}^{(t-1)}$ is the center memory received at round $t$.
Some common assumptions based on existing frameworks are reasonably presented below.

\begin{assumption}[$L_M$-smooth]\label{assump:smoothM}
Within a fixed communication round $t$, treat $\mathbf{B}^{(t-1)}$ as a constant. Excluding the measure-zero non-differentiable set,
\begin{equation}
\|\nabla_{\mathbf{M}}\ell(\mathbf{M},\mathbf{B})-\nabla_{\mathbf{M}}\ell(\tilde{\mathbf{M}},\mathbf{B})\|_2\le L_M\|\mathbf{M}-\tilde{\mathbf{M}}\|_2, 
\end{equation}
and $\|\nabla_{\mathbf{M}}\ell(\mathbf{M},\mathbf{B})\|_2\le G_M$.
By the chain rule, the composite smoothness constant satisfies $L_\theta\le B_J^2L_M$, and $\|\nabla_\theta \ell(\mathbf{M},\mathbf{B})\|_2\le B_JG_M$.
\end{assumption}

\begin{assumption}[Bounded features and Jacobian]\label{assump:boundedM}
There exist constants $R_M,B_J>0$ such that, for any $\theta,x$, satisfies $\|f(x;\theta)\|_2\le R_M$, $\|J_f(x;\theta)\|_2\le B_J$.

\end{assumption}

\begin{assumption}[Unbiased stochastic gradients with bounded variance]\label{assump:sgd}
The local SGD gradient estimate $g_{e}^{(t)}$ satisfies
\begin{equation}
\mathbb E\big[g_{e}^{(t)}\big]=\nabla_\theta\,\mathbb E_{x\sim\mathcal D_n}\!\big[\ell(\mathbf{M}_{e}^{(t)} ,\mathbf{B}^{(t-1)})\big],
\end{equation}
and $\mathbb E\big\|g_{e}^{(t)}-\nabla_\theta \mathbb E_x[\ell(\mathbf{M}_{e}^{(t)} ,\mathbf{B}^{(t-1)})]\big\|_2^2\ \le\ \sigma^2$.
\end{assumption}

\begin{assumption}[Lipschitz Continuity in $\mathbf{B}$]\label{assump:agg}
Client-side reduction and server aggregation are convex combinations. Together with Assumption~\ref{assump:boundedM}, this implies
$\|\mathbf{B}^{(t)}\|_2\le R_M$, hence
$\mathrm{dis}(\mathbf{m}_{q},\mathbf{b}_k)\le 2R_M$ and $0\le \ell(\cdot,\cdot)\le 2R_M$.
Moreover, the map $\ell(\mathbf{M},\cdot)$ is $1$-Lipschitz in $\mathbf{B}$:
\begin{equation}
\big|\ell(\mathbf{M},\mathbf{B})-\ell(\mathbf{M},\tilde{\mathbf{B}})\big|
\ \le\ \|\mathbf{B}-\tilde{\mathbf{B}}\|_2.
\end{equation}
This follows from the bounded convex hull, max-hinge, and the triangle inequality. Moreover, we assume that the global-memory drift induced by the rEMA-smoothed reduced memories and server aggregation is bounded by $2R_M/(t+1)$ under the bounded-memory condition.
\end{assumption}

\begin{theorem}[Round Loss Reduction]\label{thm:loss-reduction}
Within a fixed communication round $t$, under Assumptions~\ref{assump:smoothM}--\ref{assump:agg},
for any client and local iterations $e=0,\ldots,E-1$, if $\eta\in(0,2/L_\theta)$, then
\begin{equation}\label{eq:per-round}
\begin{aligned}
\mathbb E\big[\mathcal L_{E}^{(t)}\big]
&\le \mathbb E\big[\mathcal L_{0}^{(t)}\big]
-\Big(\eta-\tfrac{\eta^2L_\theta}{2}\Big)\sum_{e=0}^{E-1}\mathbb E\!\left[\|\nabla_\theta \mathcal L_{e}^{(t)}\|_2^2\right] \\
&\quad+\frac{E\,\eta^2L_\theta}{2}\,\sigma^2,
\end{aligned}
\end{equation}
and $\forall e$, $\mathbb E[\mathcal L_{e}^{(t)}]\le 2R_M<\infty$.
\end{theorem}

\begin{proof}
Assumption~\ref{assump:agg} implies $\mathrm{dis}(\mathbf{m},\mathbf{b})\le 2R_M$, hence
$\max(0,\mathrm{dis}-\tau_\ell)\le \max(0,2R_M-\tau_\ell)\le 2R_M$,
as $\tau_\ell > 0$. Therefore $\mathbb E[\mathcal L_{e}^{(t)}]\le 2R_M$.

Let the update be $\theta_{e+1}^{(t)}=\theta_{e}^{(t)}-\eta\,g_{e}^{(t)}$.
By the Assumption~\ref{assump:smoothM}, apply the one-step smooth upper bound to
$\theta\mapsto \mathbb E[\mathcal L_{e}^{(t)}]$:
\[
\begin{aligned}
\mathbb E[\mathcal L_{e+1}^{(t)}]
&\le \mathbb E[\mathcal L_{e}^{(t)}]
+\langle \nabla_\theta \mathbb E[\mathcal L_{e}^{(t)}],\theta_{e+1}^{(t)}-\theta_{e}^{(t)}\rangle\\
&\quad+\frac{L_\theta}{2}\|\theta_{e+1}^{(t)}-\theta_{e}^{(t)}\|_2^2 \\
&\le \mathbb E[\mathcal L_{e}^{(t)}]
-\eta\,\mathbb E\!\left[\langle \nabla_\theta \mathcal L_{e}^{(t)},\,g_{e}^{(t)}\rangle\right]
+\frac{\eta^2L_\theta}{2}\,\mathbb E\|g_{e}^{(t)}\|_2^2.
\end{aligned}
\]

By Assumption~\ref{assump:sgd},
$\mathbb E[g_{e}^{(t)}]=\mathbb E[\nabla_\theta \mathcal L_{e}^{(t)}]$ and
$\mathbb E\|g_{e}^{(t)}\|_2^2=\mathbb E\|\nabla_\theta \mathcal L_{e}^{(t)}\|_2^2+\sigma^2$; substituting yields
\[
\mathbb E[\mathcal L_{e+1}^{(t)}]
\le \mathbb E[\mathcal L_{e}^{(t)}]
-\Big(\eta-\tfrac{\eta^2L_\theta}{2}\Big)\mathbb E\|\nabla_\theta \mathcal L_{e}^{(t)}\|_2^2
+\frac{\eta^2L_\theta}{2}\sigma^2.
\]
Summing over $e=0,\ldots,E-1$ gives \eqref{eq:per-round}.
\end{proof}

\begin{theorem}[FedDyMem Convergence]\label{thm:convergence}
Under Assumptions~\ref{assump:smoothM}--\ref{assump:agg} with fixed local steps $E\ge 1$, as $T\to\infty$, we have
\begin{equation}
\frac{1}{TE}\sum_{t=1}^T\sum_{e=0}^{E-1}\mathbb E\Big[\|\nabla_\theta \mathcal L_{e}^{(t)}\|_2^2\Big]\ \longrightarrow\ 0.
\end{equation}
\end{theorem}

\begin{proof}
Sum Theorem~\ref{thm:loss-reduction} over $t=1,\ldots,T$ to obtain
\begin{equation}\label{eq:sum-loss-reduction}
\begin{aligned}
&\Big(\eta-\tfrac{\eta^2L_\theta}{2}\Big)\sum_{t=1}^T\sum_{e=0}^{E-1}\mathbb E\|\nabla_\theta \mathcal L_{e}^{(t)}\|_2^2 \\
&\ \le\ \sum_{t=1}^T\big(\mathbb E[\mathcal L_{0}^{(t)}]-\mathbb E[\mathcal L_{E}^{(t)}]\big) 
+T\cdot\tfrac{E\,\eta^2L_\theta}{2}\,\sigma^2.
\end{aligned}
\end{equation}
Apply a decomposition to the first term on the right:
\begin{equation}
\begin{aligned}
\sum_{t=1}^T\big(\mathbb E[\mathcal L_{0}^{(t)}]-\mathbb E[\mathcal L_{E}^{(t)}]\big)
&=\mathbb E[\mathcal L_{0}^{(1)}]-\mathbb E[\mathcal L_{E}^{(T)}]\\
&+\sum_{t=1}^{T-1}\Big(\mathbb E[\mathcal L_{0}^{(t+1)}]-\mathbb E[\mathcal L_{E}^{(t)}]\Big),
\end{aligned}
\end{equation}
Let $\Delta_t:=\mathbb E\|\mathbf{B}^{(t)}-\mathbf{B}^{(t-1)}\|_2$. By the $1$-Lipschitz property in Assumption~\ref{assump:agg},
\begin{equation}
\mathbb E[\mathcal L_{0}^{(t+1)}]-\mathbb E[\mathcal L_{E}^{(t)}]
 \le\ \Delta_t.
\end{equation}
Using the lower bound $\mathbb E[\mathcal L_{E}^{(T)}]\ge \underline\ell$, substitute back into \eqref{eq:sum-loss-reduction} and divide both sides by $TE(\eta-\eta^2L_\theta/2)$ to get
\begin{equation}\label{eq:ergodic-bound}
\begin{aligned}
&\frac{1}{TE}\sum_{t=1}^T\sum_{e=0}^{E-1}
\mathbb E\!\left[\|\nabla_\theta \mathcal L_{e}^{(t)}\|_2^2\right]
\le \frac{\mathbb E[\mathcal L_{0}^{(1)}]-\underline\ell}{(\eta-\eta^2L_\theta/2)\,TE} \\
&+\frac{\eta L_\theta}{2(\eta-\eta^2L_\theta/2)}\,\sigma^2
+\frac{1}{(\eta-\eta^2L_\theta/2)\,TE}\sum_{t=1}^{T-1}\Delta_t.
\end{aligned}
\end{equation}

For the first term, with $\eta\in(0,2/L_\theta)$ and $E\ge 1$, the numerator is constant while the denominator grows linearly in $T$, hence it vanishes. For the second term, either use decreasing step sizes, leading to a noise factor proportional to $(\sum_{t,e}(\eta_{e}^{(t)})^2)/(\sum_{t,e}\eta_{e}^{(t)})\to 0$ if $\sum\eta_{e}^{(t)}=\infty$ and $\sum(\eta_{e}^{(t)})^2<\infty$, or fix $\eta$ and increase batch size so that $\sigma^2\to 0$; in both cases the term vanishes. For the third term, using  Assumption~\ref{assump:agg}, we have
\[
\Delta_t
\le\ \frac{2R_M}{t+1}.
\]
Thus $\sum_{t=1}^{T-1}\Delta_t=O(\log T)$, and after normalization by $TE$ it tends to $0$.
Combining the three parts proves the claim.
\end{proof}

The above analysis shows that the expected loss-gradient converges to zero as iterations increase, thereby proving FedDyMem’s convergence.

\section{Experiments}
All experiments were implemented in PyTorch and conducted on a server equipped with AMD EPYC 9554 CPUs and NVIDIA GeForce RTX A6000 GPUs.

\subsection{Experimental Settings}
\noindent \textbf{Federated Tasks and Data Sources.} We evaluate FedDyMem on six independently constructed federated UAD tasks built from eleven public industrial and medical datasets. Each task corresponds to a separate federation with a shared inspection or screening objective. The six tasks are summarized as:
\begin{itemize}
    \item \textbf{Brain-AD Task(Brain):} We combined $3$ brain image diagnostic datasets, including Brain-Tumor (\cite{cai2024medianomaly}), BraTS2021 (\cite{baid2021rsna}) and HeadCT (\cite{salehi2021multiresolution}). These datasets, consisting of both MRI and CT images, were distributed across different clients in a heterogeneous manner to represent a mixture of various image types.
    \item \textbf{ChestXray-AD Task(X-ray):} Three diagnostic chest X-ray datasets, including ChestXRay2017 (\cite{kermany2018identifying}), RSNA (\cite{wang2017chestx}), and VinDr-CXR (\cite{nguyen2022vindr}), were used to construct the ChestXray-AD task. These datasets were derived from distinct medical institutions, ensuring a diverse and heterogeneous data distribution, which is essential for the federated UAD.
    \item \textbf{PCB-AD Task(PCB):} We construct the PCB-AD task by combining four PCB product types from VisA (\cite{zou2022spot}) and two PCB product types from VISION (\cite{bai2023vision}). Consequently, the dataset provides a total of $6$ distinct types of PCB product images, facilitating simulation in federated UAD.
    \item \textbf{Texture-AD Task(Texture):} This task utilizes $5$ texture classes(carpet, grid, leather, tile, and wood) from the MVTec dataset (\cite{bergmann2019mvtec}), which includes mask annotations, to simulate industrial federated anomaly detection scenarios.
    \item \textbf{MetalParts-AD Task(MPDD) (\cite{jezek2021deep}):} The MPDD comprises images of $6$ types of metal parts, each captured under varying conditions of spatial orientation, object positioning, and camera distance. These images are further characterized by diverse lighting intensities and non-homogeneous backgrounds, providing a comprehensive dataset for our federated UAD.
    \item \textbf{Wafer-AD Task(Wafer) (\cite{lei2025adapted}):} Wafer-AD task comprises a dataset of 14 distinct wafer products, each imaged using an industrial-grade, high-resolution optical camera. This dataset also explores these wafer products under various optical conditions, thus providing a rich federated UAD simulate environment with diverse feature distributions. 
\end{itemize} 
\begin{table}[t]
\footnotesize
\centering
\caption{Default non-IID client partition used in the main experiments.}
\label{tab:partition}
\begin{tabular}{c|c|c c c}
\toprule
Task & \makecell{Partitioning\\attribute}  & \#Groups & $N$ & $\alpha$ \\
\midrule
Brain & Data source & 3 & 5 & 0.1 \\
X-ray & Data source & 3 & 5 & 0.1 \\
PCB & PCB product type & 6 & 5 & 0.1 \\
Texture & Texture class & 5 & 5 & 0.1 \\
MPDD & Metal-part type & 6 & 5 & 0.1 \\
Wafer & Wafer product type & 14 & 5 & 0.1 \\
\bottomrule
\end{tabular}
\end{table}

\begin{figure*}[t]
\centering
\includegraphics[width=\linewidth]{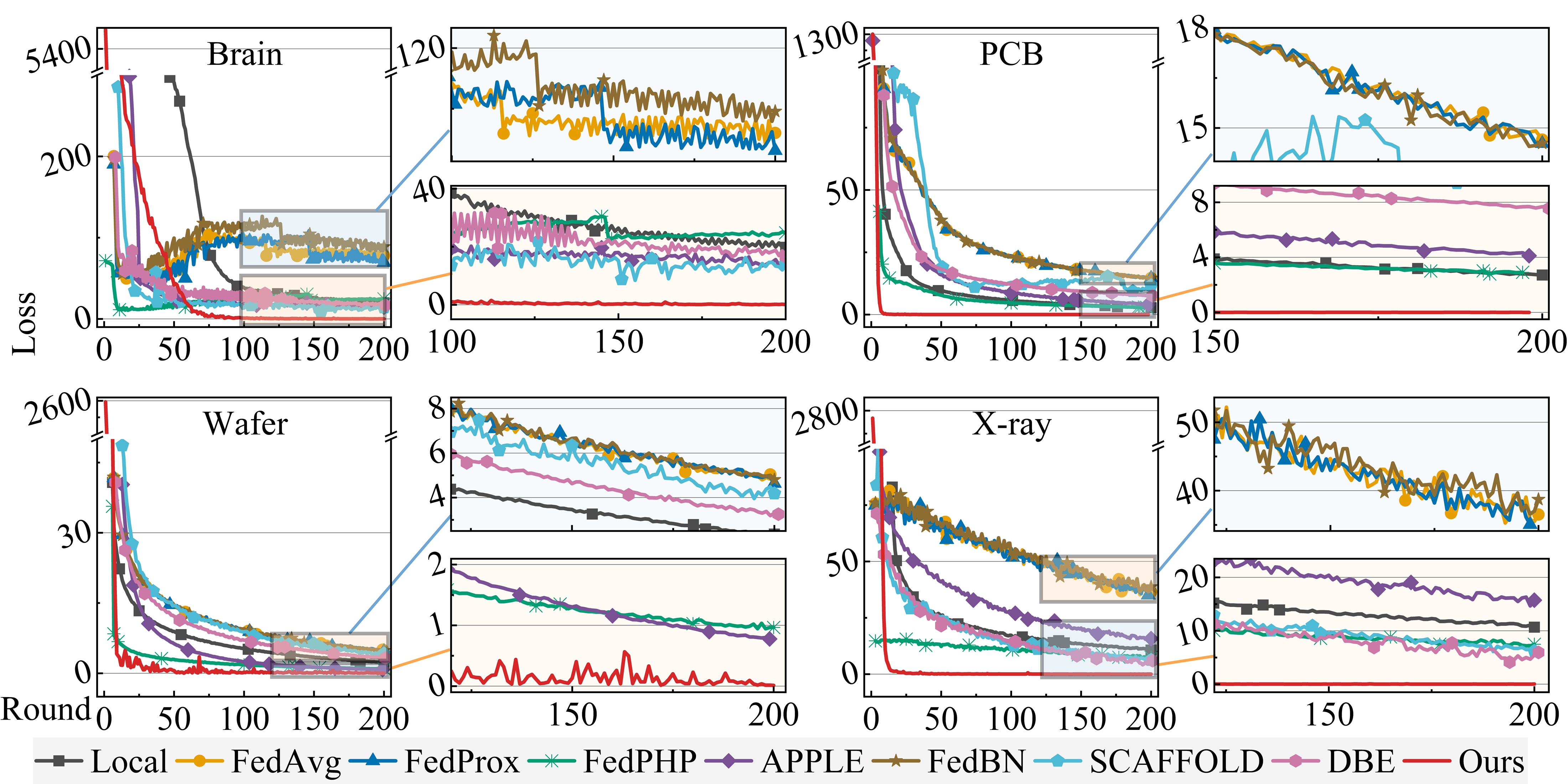}
\caption{Local metric loss curves of the first client over $200$ communication rounds on four federated UAD tasks.}
\label{exp-loss}
\end{figure*}

\noindent \textbf{Client Data Partition.} For each federated task, we construct non-IID client datasets from normal training samples according to task-specific source or subtype attributes. Specifically, the normal samples are first grouped by the corresponding data source, product type, texture class, metal-part type, or wafer product type, as summarized in Table~\ref{tab:partition}. The main experiments use $5$ simulated clients, and the samples in each group are allocated among them according to a Dirichlet distribution with $\alpha_\mathrm{Dir}=0.1$. This default partition produces non-IID client datasets in which all training samples are normal, while the proportions of different groups vary across clients. The impact of alternative client partition configurations is further examined in Section~\ref{robustness}.

\noindent \textbf{Implementation Details.} We use a pretrained WideResNet-50 (\cite{zagoruyko2016wide}) as the frozen backbone and extract features from $L=4$ layers. The projection layer, memory generator, and trainable grid are optimized locally, while the backbone remains fixed throughout training. The default trainable grid size is set to $H_g=W_g=8$, the margin in the local metric loss is set to $\tau_\ell=0.01$, and the number of nearest neighbors is set to $k_{\mathrm{nn}}=3$. Each uploaded reduced memory bank contains $Q=H_1W_1$ prototypes of dimension $C$. The server aggregates them into $K_{\mathrm{clu}}=Q$ global prototypes with a $I_{\mathrm{clu}}=15$ iterations in each aggregation round. In default setting, all five simulated clients participate in every communication round. We train for $T=200$ communication rounds, with $E=1$ local epoch and a batch size of $10$. The trainable local modules are optimized using Adam with a learning rate of $1\times10^{-3}$ and weight decay of $5\times10^{-4}$. The clients share the same initialization of model parameters. All images in the dataset were preprocessed by cropping and resizing to a resolution of $224\times 224$ pixels using the bicubic interpolation method. All compared federated methods use the same client partitions, communication budget, local training schedule, and input preprocessing and their aggregation procedures follow their respective algorithms.

\noindent \textbf{Evaluation Metric.} For the evaluation of anomaly detection, we adopt the Area Under the Receiver Operating Characteristic Curve (AUROC), ensuring a comprehensive assessment of our proposed model in alignment with prior studies. Image-level anomaly detection performance is quantified using the standard AUROC metric, denoted as Image-level AUROC(I-AUROC). Additionally, to assess anomaly localization, we compute the Pixel-level AUROC(P-AUROC). To further enhance the evaluation of anomaly localization capabilities, we calculate the Per-Region Overlap (PRO) at the pixel level, providing a more granular and comprehensive analysis. These metrics are widely used in previous work (\cite{roth2022towards,you2022unified,cai2024medianomaly}).

\begin{figure*}[t]
\centering
\includegraphics[width=0.98\linewidth]{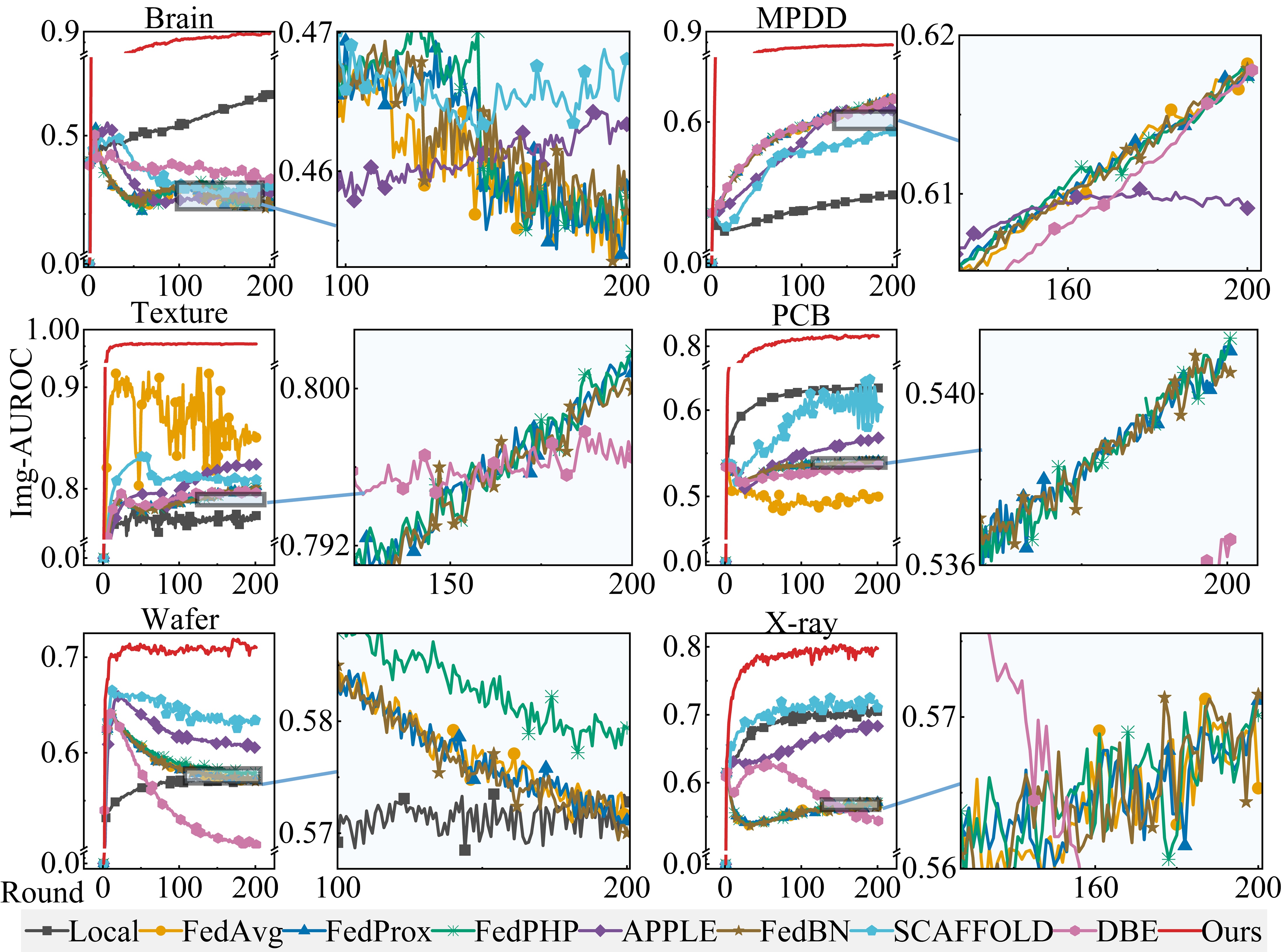} 
\caption{Illustration of I-AUROC over communication rounds on six federated UAD tasks.}
\label{exp-acc}
\end{figure*}

\begin{figure*}[t]
\centering
\includegraphics[width=\linewidth]{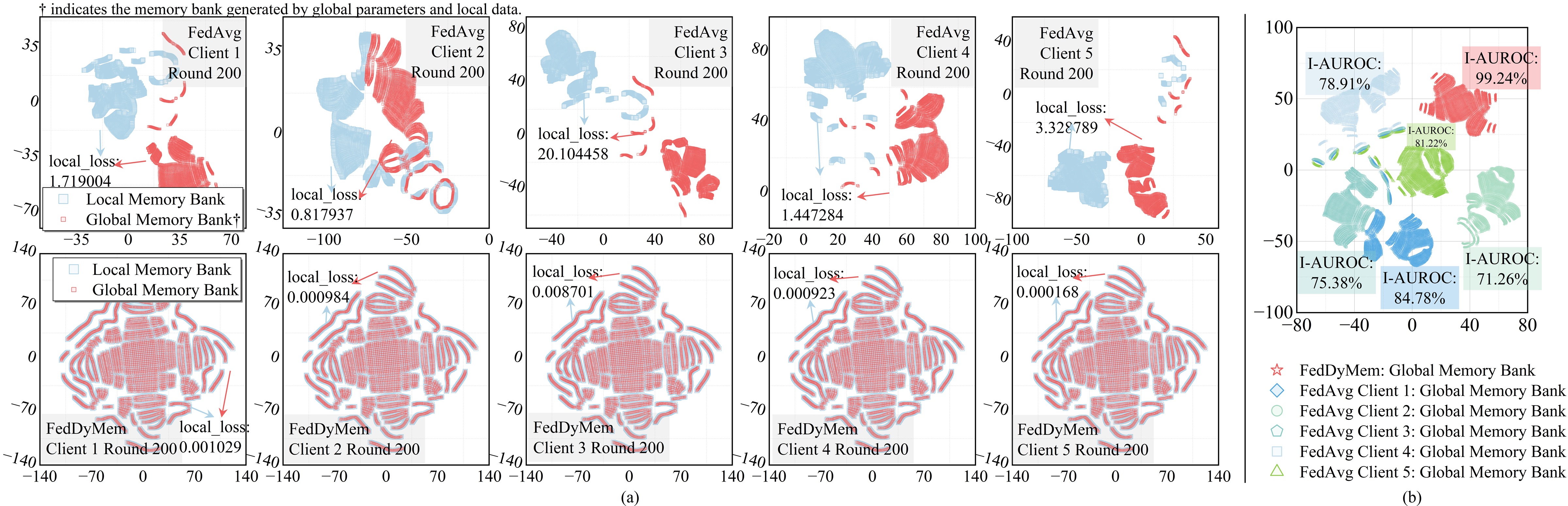}
\caption{Comparison of feature distributions on the texture task after 200 communication rounds. (a) t-SNE (\cite{maaten2008visualizing}) visualization of local and global memory banks generated by five clients using FedAvg (top row) and FedDyMem (bottom row). (b) Comparison of global memory banks aggregated from FedDyMem and FedAvg clients, together with the corresponding I-AUROC results.}
\label{loss-tsne}
\end{figure*}

\begin{table*}[t]
\centering
\setlength{\tabcolsep}{1.0mm}{
\footnotesize
\caption{Comparison of Image-AUROC (\%) results on six federated UAD tasks under the default non-IID protocol.}
\label{tab:comparison}
\begin{tabular}{l|c|cccccc}
\toprule
Methods & Communication
& Wafer & PCB & Texture & MPDD & Brain & X-ray \\
\midrule

Local
& -
& 57.35($\pm$0.71)
& 51.89($\pm$0.76)
& 75.94($\pm$0.39)
& 54.39($\pm$0.77)
& 50.66($\pm$0.87)
& 61.02($\pm$0.58) \\
\midrule

FedAvg
& \multirow{7}{*}{
    \makecell{
        Model Parameters: \\
        $\approx$ 10.66M, \\
        (SCAFFOLD) $\approx$ \\
        10.66M $\times$ 2
    }
}
& 63.91($\pm$0.36)
& 54.18($\pm$0.33)
& 80.32($\pm$0.81)
& 61.87($\pm$0.76)
& 50.16($\pm$0.63)
& 61.48($\pm$0.56) \\

FedProx
&
& 63.98($\pm$0.51)
& 54.11($\pm$0.48)
& 80.22($\pm$0.78)
& 61.83($\pm$0.61)
& 50.47($\pm$0.34)
& 61.45($\pm$0.66) \\

FedPHP
&
& 63.94($\pm$0.56)
& 54.13($\pm$0.36)
& 80.19($\pm$0.62)
& 61.80($\pm$0.63)
& 50.38($\pm$0.39)
& 61.45($\pm$0.37) \\

APPLE
&
& 65.92($\pm$0.33)
& 56.79($\pm$0.65)
& 82.43($\pm$0.51)
& 61.03($\pm$0.77)
& 50.82($\pm$0.45)
& 68.47($\pm$0.51) \\

FedBN
&
& 63.89($\pm$0.35)
& 54.09($\pm$0.79)
& 80.13($\pm$0.77)
& 61.90($\pm$0.61)
& 50.04($\pm$0.54)
& 61.45($\pm$0.72) \\

SCAFFOLD
&
& 66.53($\pm$0.85)
& 63.56($\pm$0.30)
& 83.26($\pm$0.37)
& 59.41($\pm$0.80)
& 49.90($\pm$0.52)
& 72.52($\pm$0.52) \\

DBE
&
& 64.12($\pm$0.49)
& 53.78($\pm$0.38)
& 79.78($\pm$0.90)
& 61.78($\pm$0.51)
& 50.06($\pm$0.85)
& 62.77($\pm$0.61) \\

\midrule

\textbf{Ours}
& \makecell{
    Memory Size: $\approx$ \\
    \textbf{5.62M}
}
& \textbf{71.96($\pm$0.32)}
& \textbf{83.34($\pm$0.45)}
& \textbf{99.24($\pm$0.36)}
& \textbf{83.93($\pm$0.77)}
& \textbf{89.27($\pm$0.45)}
& \textbf{80.30($\pm$0.36)} \\
\bottomrule
\end{tabular}
}
\end{table*}

\subsection{Convergence Curves and Distribution Consistency}

Fig.~\ref{exp-loss} compares the training loss curves of FedDyMem and baseline methods on Brain, Wafer, PCB, and X-ray over $200$ communication rounds. All methods use the same local metric loss and training configuration for fair comparison. FedDyMem converges faster and reaches lower loss values across the datasets. In the later rounds, its loss curves remain more stable with smaller fluctuations, which is consistent with the convergence analysis in Section~\ref{convergence}. Fig.~\ref{exp-acc} further reports I-AUROC over communication rounds on all six federated UAD tasks, showing that FedDyMem achieves higher anomaly detection performance than the baselines.

Fig.~\ref{loss-tsne} visualizes memory representations on the Texture task after $200$ rounds. In Fig.~\ref{loss-tsne}(a), FedAvg local memories show clear deviations from the global memory, with local losses up to $10^4$ times higher than FedDyMem. This indicates that FedAvg struggles to align local memory features with the global distribution under client heterogeneity. In contrast, Fig.~\ref{loss-tsne}(b) shows that FedDyMem forms a more compact and coherent global feature space. The dynamic memory update and $K$-means aggregation help align client representations and reduce the effect of intra-class distribution bias.

\begin{table*}[t]
\centering
\setlength{\tabcolsep}{0.2mm}{
\footnotesize
\caption{Comparison of Pixel-AUROC (\%) and PRO (\%) results on four federated UAD tasks under the default non-IID protocol.}
\label{tab:compare_loc_metrics}
\begin{tabular}{l|cccc|cccc}
\toprule
\multirow{2}{*}{Methods}
& \multicolumn{4}{c|}{Pixel-AUROC (\%)}
& \multicolumn{4}{c}{PRO (\%)} \\
\cmidrule{2-9}
& Wafer & PCB & Texture & MPDD
& Wafer & PCB & Texture & MPDD \\
\midrule
Local
& 58.27($\pm$1.20)
& 68.12($\pm$1.11)
& 72.63($\pm$0.98)
& 87.02($\pm$1.05)
& 47.76($\pm$0.52)
& 45.86($\pm$1.31)
& 56.97($\pm$0.48)
& 51.31($\pm$0.73) \\
\midrule
FedAvg
& 66.56($\pm$0.70)
& 69.30($\pm$0.34)
& 76.16($\pm$0.73)
& 91.67($\pm$0.94)
& 67.00($\pm$0.30)
& 54.04($\pm$0.57)
& 53.35($\pm$1.49)
& 55.07($\pm$1.25) \\

FedProx
& 66.41($\pm$0.74)
& 69.10($\pm$1.28)
& 76.08($\pm$1.28)
& 91.59($\pm$1.02)
& 67.25($\pm$0.33)
& 54.04($\pm$1.28)
& 53.62($\pm$0.70)
& 54.96($\pm$1.16) \\

FedPHP
& 66.44($\pm$1.38)
& 69.11($\pm$0.42)
& 76.03($\pm$0.57)
& 91.56($\pm$0.73)
& 67.16($\pm$0.38)
& 54.23($\pm$1.16)
& 53.37($\pm$0.47)
& 54.86($\pm$1.42) \\

APPLE
& 69.88($\pm$0.95)
& 71.34($\pm$0.85)
& 78.06($\pm$1.08)
& 92.12($\pm$1.19)
& 69.02($\pm$0.42)
& 60.11($\pm$1.45)
& 54.14($\pm$1.16)
& 55.84($\pm$0.76) \\

FedBN
& 66.32($\pm$0.57)
& 68.91($\pm$0.82)
& 75.92($\pm$1.34)
& 91.59($\pm$0.38)
& 67.16($\pm$0.68)
& 54.21($\pm$1.07)
& 53.48($\pm$0.79)
& 54.92($\pm$0.82) \\

SCAFFOLD
& 68.75($\pm$1.19)
& 77.18($\pm$0.54)
& 80.81($\pm$0.35)
& 93.47($\pm$0.74)
& 68.74($\pm$1.03)
& 66.43($\pm$0.58)
& 55.79($\pm$1.30)
& 58.71($\pm$0.84) \\

DBE
& 66.95($\pm$1.07)
& 72.95($\pm$1.35)
& 76.18($\pm$0.89)
& 92.66($\pm$1.36)
& 68.46($\pm$1.05)
& 57.89($\pm$0.75)
& 53.86($\pm$0.58)
& 55.76($\pm$0.95) \\

\midrule
\textbf{Ours}
& \textbf{78.40($\pm$1.31)}
& \textbf{96.03($\pm$0.79)}
& \textbf{97.04($\pm$0.56)}
& \textbf{98.97($\pm$1.39)}
& \textbf{74.91($\pm$0.54)}
& \textbf{81.09($\pm$1.31)}
& \textbf{86.81($\pm$1.13)}
& \textbf{79.43($\pm$1.15)} \\
\bottomrule
\end{tabular}
}
\end{table*}

\begin{table*}[t]
\centering
\setlength{\tabcolsep}{0.40mm}{
\footnotesize
\caption{Comparison with representative UAD methods under centralized and federated settings. Results are reported as the mean I-AUROC/P-AUROC/PRO (\%).}
\label{tab:representative_uad_comparison}
\begin{tabular}{l|c|c|c|c|c|c|c}
\toprule
Methods & Setting
& Wafer
& PCB
& Texture
& MPDD
& Brain
& X-ray \\
\midrule

\multicolumn{8}{l}{
\textit{All normal training data are jointly accessible.}
} \\
\midrule

PatchCore
& Centralized
& \makecell{75.56/77.92/71.80}
& \makecell{91.25/96.12/77.97}
& \makecell{99.04/97.46/88.76}
& \makecell{85.00/98.88/71.31}
& \makecell{95.69/-/-}
& \makecell{78.10/-/-} \\

SimpleNet
& Centralized
& \makecell{69.92/72.16/70.51}
& \makecell{84.35/96.17/70.47}
& \makecell{99.32/96.43/85.59}
& \makecell{86.91/98.92/74.65}
& \makecell{92.79/-/-}
& \makecell{77.96/-/-} \\

UniAD
& Centralized
& \makecell{74.05/83.23/76.56}
& \makecell{83.79/96.74/79.59}
& \makecell{98.92/94.90/89.63}
& \makecell{69.28/87.65/70.06}
& \makecell{87.51/-/-}
& \makecell{\textbf{84.89}/-/-} \\

MambaAD
& Centralized
& \makecell{76.34/84.69/84.67}
& \makecell{94.59/\textbf{99.11}/\textbf{91.52}}
& \makecell{98.41/97.99/88.15}
& \makecell{88.11/96.74/90.58}
& \makecell{97.73/-/-}
& \makecell{82.67/-/-} \\

InpFormer
& Centralized
& \makecell{\textbf{81.85}/\textbf{85.02}/\textbf{85.05}}
& \makecell{\textbf{94.74}/98.87/88.36}
& \makecell{\textbf{99.82}/\textbf{98.67}/\textbf{95.56}}
& \makecell{\textbf{91.09}/\textbf{99.55}/\textbf{94.56}}
& \makecell{\textbf{98.73}/-/-}
& \makecell{83.44/-/-} \\

Ours$^{\dagger}$
& Centralized
& \makecell{73.42/79.12/75.67}
& \makecell{84.45/96.72/82.12}
& \makecell{99.02/97.66/87.21}
& \makecell{85.73/95.29/80.14}
& \makecell{90.36/-/-}
& \makecell{81.11/-/-} \\

\midrule
\multicolumn{8}{l}{
\textit{Raw normal training data remain client-local.}
} \\
\midrule
PatchCore
& \makecell{-}
& \makecell{68.17/69.03/68.34}
& \makecell{71.59/91.06/65.47}
& \makecell{90.31/86.01/82.53}
& \makecell{69.55/84.41/59.76}
& \makecell{71.16/-/-}
& \makecell{68.22/-/-} \\

SimpleNet
& FedAvg
& \makecell{64.28/63.64/67.68}
& \makecell{65.03/83.13/62.90}
& \makecell{88.20/85.79/71.63}
& \makecell{69.02/88.75/63.70}
& \makecell{69.27/-/-}
& \makecell{67.58/-/-} \\

UniAD
& FedAvg
& \makecell{65.33/68.02/67.57}
& \makecell{62.25/90.08/64.04}
& \makecell{88.39/84.97/80.34}
& \makecell{63.18/85.85/56.44}
& \makecell{73.93/-/-}
& \makecell{68.70/-/-} \\

MambaAD
& FedAvg
& \makecell{69.77/72.33/70.67}
& \makecell{73.47/91.13/73.25}
& \makecell{93.48/92.88/81.63}
& \makecell{77.12/90.47/66.38}
& \makecell{86.82/-/-}
& \makecell{75.85/-/-} \\

InpFormer
& FedAvg
& \makecell{70.98/74.05/72.90}
& \makecell{78.81/93.26/77.78}
& \makecell{97.48/92.88/83.63}
& \makecell{80.20/94.36/75.13}
& \makecell{88.90/-/-}
& \makecell{78.90/-/-} \\

Ours$^{\dagger}$
& FedAvg
& \makecell{63.91/66.56/67.00}
& \makecell{54.18/69.30/54.04}
& \makecell{80.32/76.16/53.35}
& \makecell{61.87/91.67/55.07}
& \makecell{50.16/-/-}
& \makecell{61.48/-/-} \\

\midrule
Ours
& FedDyMem
& \makecell{\textbf{71.96}/\textbf{78.40}/\textbf{74.91}}
& \makecell{\textbf{83.34}/\textbf{96.03}/\textbf{81.09}}
& \makecell{\textbf{99.24}/\textbf{97.04}/\textbf{86.81}}
& \makecell{\textbf{83.93}/\textbf{98.97}/\textbf{79.43}}
& \makecell{\textbf{89.27}/-/-}
& \makecell{\textbf{80.30}/-/-} \\

\bottomrule
\multicolumn{8}{l}{
\footnotesize
$^{\dagger}$ Reference variants of the FedDyMem architecture, which use the same underlying trainable modules as FedDyMem.
}
\end{tabular}
}
\end{table*}

\subsection{Performance Comparisons with Existing Methods}
\noindent \textbf{Federated Comparison.} To demonstrate the superior performance of FedDyMem, we conducted extensive comparative analyses with $8$ baseline models. These include the basic federated learning algorithm FedAvg (\cite{fedavg}), the regularization-based FedProx (\cite{li2020federated2}),  the personalized-based FedPHP (\cite{li2021fedphp}) and APPLE (\cite{luo2022adapt}). Additionally, we compared against domain-skew-oriented methods such as FedBN (\cite{li2021fedbn}) and SCAFFOLD (\cite{karimireddy2020scaffold}), as well as the model-splitting method DBE (\cite{feddbe}) and a purely local training scheme (Local). 
Table~\ref{tab:comparison} reports image-level anomaly detection results. Conventional federated methods, including FedAvg, FedProx, FedPHP, and DBE, show similar I-AUROC performance with limited or inconsistent gains over Local training under heterogeneous federated settings. FedDyMem achieves the highest I-AUROC across all six tasks, reaching $71.96\%$ on Wafer, $83.34\%$ on PCB, $99.24\%$ on Texture, $83.93\%$ on MPDD, $89.27\%$ on Brain, and $80.30\%$ on X-ray. It also reduces communication cost by exchanging a reduced memory bank with about $5.62$M elements, compared with about $10.66$M parameters in conventional parameter-sharing methods. Under the same numerical precision, this gives a $47.3\%$ reduction in the nominal size of the primary exchanged representation.

Table~\ref{tab:compare_loc_metrics} reports anomaly localization results on datasets with pixel-level annotations. FedDyMem achieves the highest mean P-AUROC and PRO on Wafer, PCB, Texture, and MPDD. Specifically, it obtains P-AUROC scores of $78.40\%$, $96.03\%$, $97.04\%$, and $98.97\%$, with PRO scores of $74.91\%$, $81.09\%$, $86.81\%$, and $79.43\%$, respectively.

\noindent \textbf{Representative UAD Comparison.} We further compare FedDyMem with representative UAD methods, including PatchCore (\cite{roth2022towards}), SimpleNet (\cite{liu2023simplenet}), UniAD (\cite{you2022unified}), MambaAD (\cite{he2024mambaad}), and InpFormer (\cite{luo2025exploring}), together with reference variants of the FedDyMem architecture, in both centralized and federated settings. In the centralized setting, all normal training samples from the participating clients are jointly available for training. For the federated setting, SimpleNet, UniAD, MambaAD, and InpFormer are optimized using FedAvg under the same client partition and communication schedule. PatchCore uses a frozen feature extractor and a non-parametric memory bank, and thus does not support iterative parameter aggregation through FedAvg. Therefore, each client extracts patch features once from its local normal samples, and the server aggregates the local patch features to construct a global memory bank for anomaly detection. Table~\ref{tab:representative_uad_comparison} reports the comparison results. The two Ours$^{\dagger}$ variants use the same modules as the local clients in FedDyMem. In the centralized setting, PatchCore provides a patch feature reference, SimpleNet represents a single-model UAD method, and UniAD, MambaAD, and InpFormer represent unified UAD models. For our centralized variant, all normal training samples are pooled and optimized jointly without client partitioning or communication. The memory is updated iteratively within this centralized training process, and the memory from the previous memory-update iteration is used as the reference. Under joint data access, MambaAD and InpFormer obtain strong centralized performance, confirming the effectiveness of unified UAD models when heterogeneous normal data are jointly observable during optimization. However, their performance is substantially reduced when moving from centralized joint optimization to the federated FedAvg setting. Averaged over the six tasks, UniAD, MambaAD, and InpFormer lose 12.78, 10.22, and 9.07 I-AUROC points, respectively. This systematic gap indicates that  unified modeling capacity alone is insufficient to reproduce centralized joint optimization in the federated setting, where each local update is optimized on a client-specific normal distribution and parameter averaging is not explicitly designed to preserve coverage of complementary normal modes across clients. The same-architecture comparison further supports this observation. Our centralized architecture loses 23.70 I-AUROC points on average under FedAvg, consistent with the difficulty of averaging locally optimized modules across heterogeneous client distributions. Although clients start from the same initialization, non-IID local updates can specialize the projection, memory-generator, and grid modules to different normal modes. Averaging these nonlinear modules does not necessarily preserve the normal representation geometry learned by individual clients, which can result in incomplete federation-wide normal-mode coverage. Unlike this Ours$^{\dagger}$+FedAvg ablation, the decentralized PatchCore reference uses a frozen shared feature extractor and aggregates local patch representations in a common pretrained feature space, thereby avoiding the averaging of client-specialized trainable modules. By keeping these modules client-specific and sharing reduced normal memories through server-side aggregation, FedDyMem narrows the centralized-to-federated gap to only 1.01 points on average while achieving the best federated results across all six image-level tasks and four localization tasks.

\begin{table}[t]
\centering
\setlength{\tabcolsep}{1mm}{
\footnotesize
\caption{Comparison under the group-disjoint federated setting. Results are reported as the mean I-AUROC/P-AUROC/PRO (\%).}
\label{tab:group_disjoint}
\begin{tabular}{c|c|c|c}
\toprule
Methods & Setting
& \makecell{Wafer\\($N=14$)}
& \makecell{Brain\\($N=3$)} \\
\midrule

UniAD
& FedAvg
& \makecell{67.48/76.06/61.75}
& \makecell{74.81/-/-} \\

MambaAD
& FedAvg
& \makecell{65.91/69.32/55.66}
& \makecell{86.42/-/-} \\

InpFormer
& FedAvg
& \makecell{64.15/71.16/53.45}
& \makecell{83.31/-/-} \\

Ours
& FedDyMem
& \makecell{\textbf{70.20}/\textbf{77.62}/\textbf{73.37}}
& \makecell{\textbf{88.24}/-/-} \\

\bottomrule
\end{tabular}
}
\end{table}

\noindent \textbf{Group-Disjoint Federated Comparison.}
To further examine federated UAD when normal modes are not jointly observable during local training, we construct a stricter group-disjoint setting on Wafer and Brain. Specifically, each wafer product type is assigned exclusively to one client, resulting in $N=14$, while each of the three data sources in Brain is assigned to one client, resulting in $N=3$. Thus, each client observes only one task-specific group during training, while the federation collectively covers all groups. Unlike centralized unified UAD, where normal samples from different groups can jointly participate in optimization, local updates in this setting cannot directly access normal modes held by other clients.

Table~\ref{tab:group_disjoint} compares FedDyMem with representative unified UAD methods adapted using FedAvg. FedDyMem achieves the best I-AUROC on both tasks, improving over the strongest FedAvg baseline by 2.72 points on Wafer and 1.82 points on Brain. On Wafer, it also improves PRO from 61.75\% to 73.37\%. Together with the centralized-to-federated comparison in Table~\ref{tab:representative_uad_comparison}, these results provide further empirical support that the challenge considered in federated UAD arises not merely from modeling multiple normal groups, but from integrating complementary normal modes that are distributed across clients and are not jointly observable during local optimization.

\begin{figure}[t]
\centering
\includegraphics[width=0.95\linewidth]{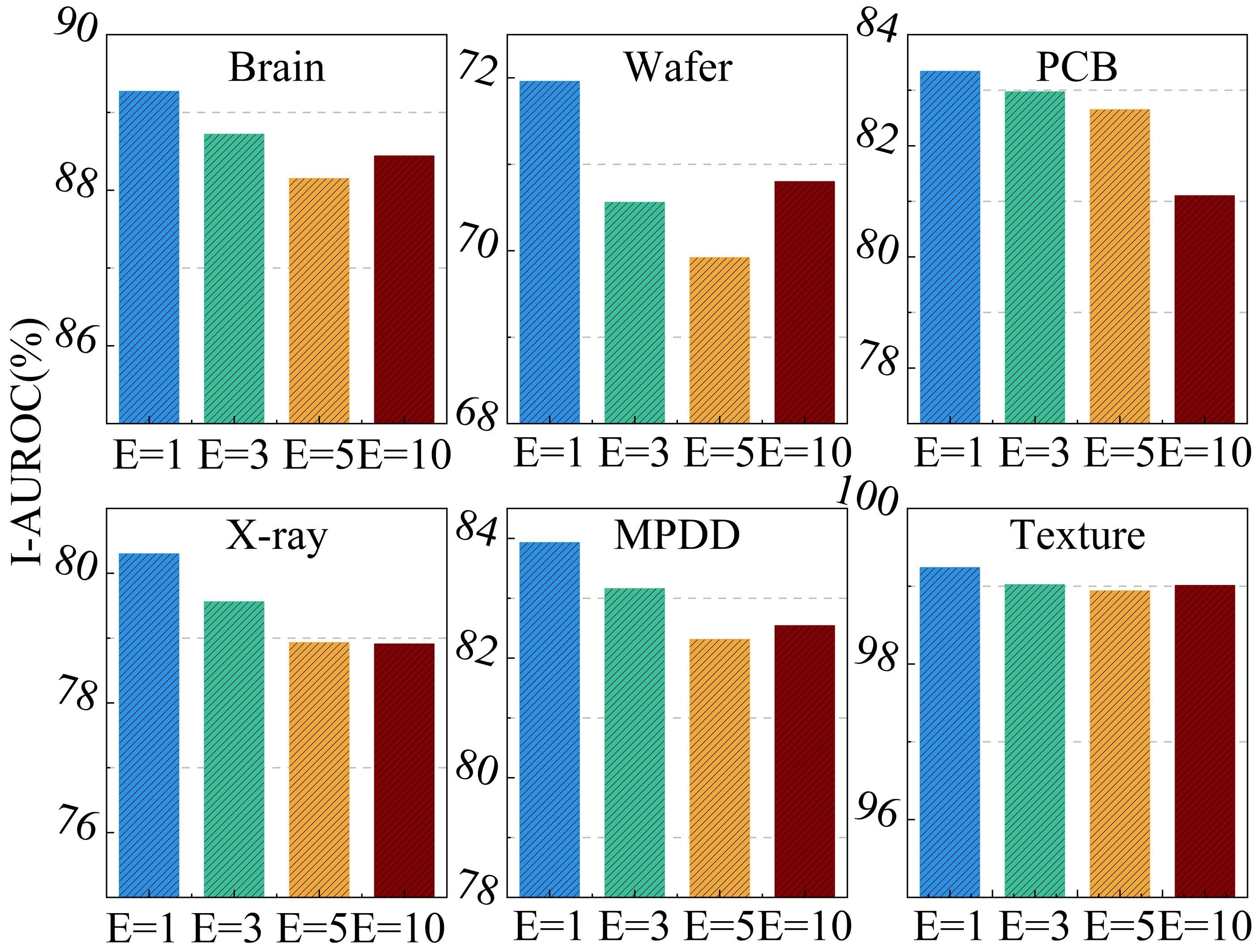} 
\caption{Impact of different local epochs on clients.}
\label{ablation-epochs}
\end{figure}

\begin{table}[t]
\centering
\setlength{\tabcolsep}{0.1mm}{
\footnotesize
\caption{Impact of per-round local training sample ratio. Results are reported as the mean I-AUROC/P-AUROC/PRO (\%).}
\label{tab:local_data_sampling}
\begin{tabular}{c|c|c|c|c}
\toprule
Ratio & Methods & Brain & MPDD & Texture \\
\midrule

\multirow{2}{*}{$20\%$}
& FedAvg
& \makecell{50.06/-/-}
& \makecell{58.65/90.72/54.22}
& \makecell{77.18/72.64/50.53} \\

& Ours
& \makecell{85.05/-/-}
& \makecell{79.60/97.97/75.68}
& \makecell{98.34/96.92/81.82} \\
\midrule

\multirow{2}{*}{$40\%$}
& FedAvg
& \makecell{50.28/-/-}
& \makecell{58.27/90.41/54.11}
& \makecell{78.50/74.16/50.20} \\

& Ours
& \makecell{86.49/-/-}
& \makecell{81.23/97.97/77.54}
& \makecell{99.09/96.89/82.53} \\
\midrule

\multirow{2}{*}{$60\%$}
& FedAvg
& \makecell{50.08/-/-}
& \makecell{60.62/91.01/54.20}
& \makecell{78.44/74.14/51.75} \\

& Ours
& \makecell{87.59/-/-}
& \makecell{82.96/97.94/78.91}
& \makecell{99.09/96.86/83.47} \\
\midrule

\multirow{2}{*}{$80\%$}
& FedAvg
& \makecell{50.16/-/-}
& \makecell{61.01/91.16/54.95}
& \makecell{80.28/75.33/52.01} \\

& Ours
& \makecell{88.47/-/-}
& \makecell{83.19/98.83/79.23}
& \makecell{99.21/96.51/85.82} \\
\midrule

\multirow{2}{*}{$100\%$}
& FedAvg
& \makecell{50.16/-/-}
& \makecell{61.87/91.67/55.07}
& \makecell{80.32/76.16/53.35} \\

& Ours
& \makecell{{89.27}/-/-}
& \makecell{{83.93}/{98.97}/{79.43}}
& \makecell{{99.24}/{97.04}/{86.81}} \\

\bottomrule
\end{tabular}
}
\end{table}

\subsection{Ablation Study}
\label{ablation}

\noindent \textbf{Impact of Local Training.} Fig.~\ref{ablation-epochs} illustrates the impact of varying the number of local training epochs ($E$) on the I-AUROC across all datasets. Across the evaluated tasks, $E=1$ achieves the highest I-AUROC. For example, in the Wafer dataset, accuracy drops from 71.96\% at $E=1$ to 69.92\% at $E=5$. Similarly, the PCB dataset demonstrates a decrease in performance, with accuracy falling from 83.34\% at $E=1$ to 82.65\% at $E=5$. MPDD, Brain, and X-ray also obtain lower I-AUROC values when using more than one local epoch. This phenomenon can be attributed to the heterogeneous data distribution across clients. As the number of local training epochs increases, the local models tend to overfit their respective client datasets, leading to a decrease in the effectiveness of the globally aggregated memory features. 

Table~\ref{tab:local_data_sampling} further evaluates whether FedDyMem remains effective when only a fraction of local normal samples is used in each communication round. With $E=1$, FedDyMem consistently outperforms FedAvg at all evaluated sample ratios. For example, its I-AUROC increases from $85.05\%$ to $88.47\%$ on Brain, from $79.60\%$ to $83.19\%$ on MPDD, and from $98.34\%$ to $99.21\%$ on Texture when the ratio increases from $20\%$ to $80\%$. Therefore, using a reduced subset of local samples can retain strong performance, whereas increasing the per-round sample coverage yields better I-AUROC in the evaluated range.

\begin{table}[t]
\centering
\setlength{\tabcolsep}{0.1mm}{
\footnotesize
\caption{Ablation study of rEMA and Dynamic Weighting (DW) in the memory-reduce.}
\label{tab:ablation_mem_reduce}
\begin{tabular}{c|c|c|c|c}
\toprule
rEMA & DW & Brain & MPDD & Texture \\
\midrule
- & -
& \makecell{84.47/-/-}
& \makecell{78.63/95.72/74.17}
& \makecell{97.91/95.32/83.13} \\

\checkmark & -
& \makecell{86.75/-/-}
& \makecell{80.48/98.88/74.85}
& \makecell{98.66/96.88/85.67} \\

- & \checkmark
& \makecell{87.62/-/-}
& \makecell{81.88/97.85/77.73}
& \makecell{98.99/96.99/85.88} \\

\checkmark & \checkmark
& \makecell{\textbf{89.27}/-/-}
& \makecell{\textbf{83.93}/\textbf{98.97}/\textbf{79.43}}
& \makecell{\textbf{99.24}/\textbf{97.04}/\textbf{86.81}} \\

\bottomrule
\end{tabular}
}
\end{table}

\begin{figure}[t]
\centering
\includegraphics[width=\linewidth]{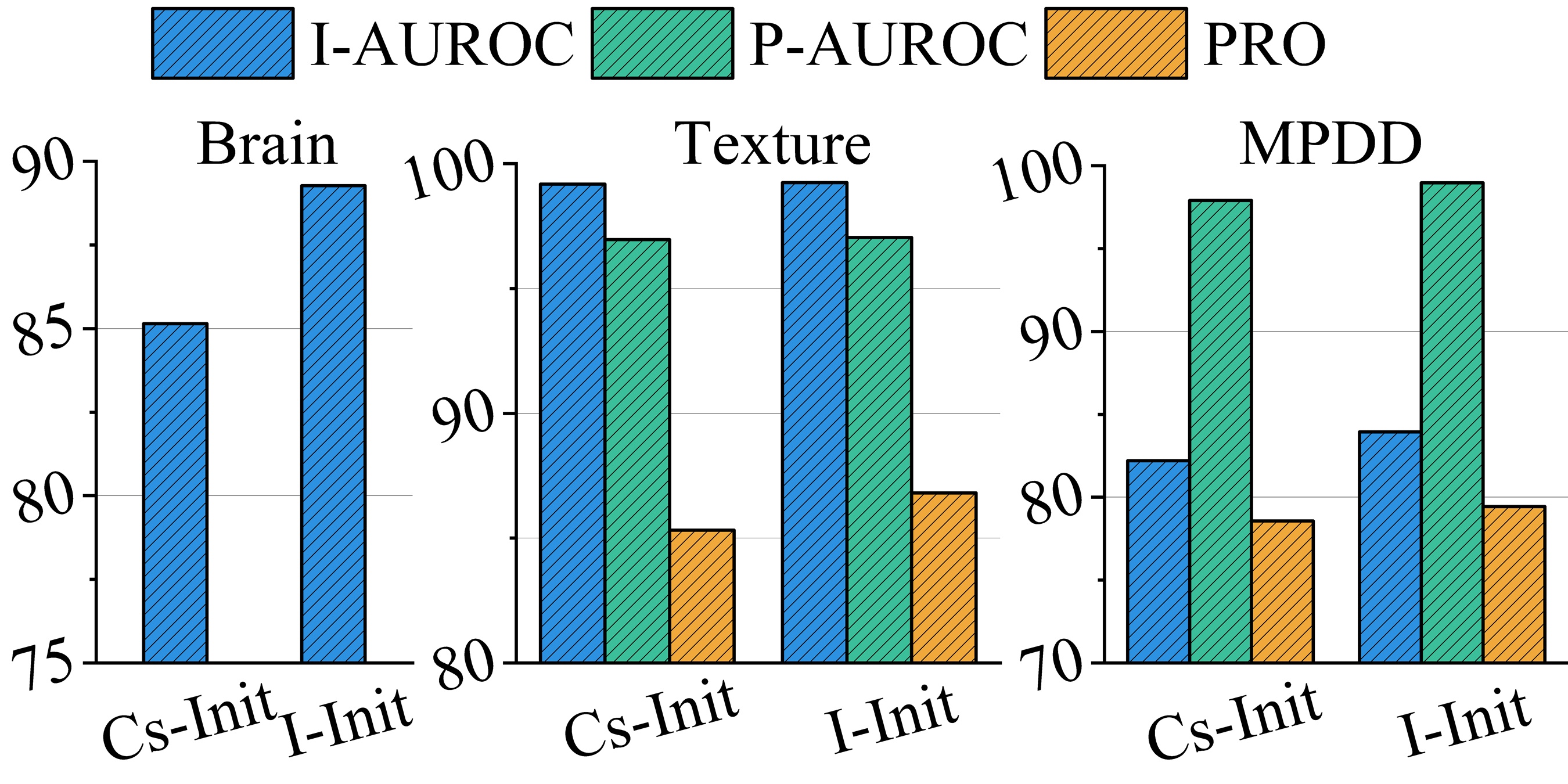}
\caption{Impact of parameter initialization.}
\label{fig:init_ablation}
\end{figure}

\begin{table}[t]
\centering
\footnotesize
\setlength{\tabcolsep}{0.8mm}{
\caption{Impact of hyperparameters and backbones. Results are reported as I-AUROC/P-AUROC/PRO (\%).}
\label{tab:ablation_hyperparameters_backbones}
\begin{tabular}{c|c|c|c}
\toprule
Factor & Setting & PCB & Brain \\
\midrule

\multirow{4}{*}{$H_g \times W_g$}
& $2 \times 2$
& \makecell{79.02/93.38/78.89}
& \makecell{86.27/-/-} \\
& $4 \times 4$
& \makecell{82.24/95.74/80.96}
& \makecell{88.73/-/-} \\
& $8 \times 8$
& \makecell{83.34/\textbf{96.03}/\textbf{81.09}}
& \makecell{\textbf{89.27}/-/-} \\
& $16 \times 16$
& \makecell{\textbf{83.36}/95.87/80.99}
& \makecell{88.93/-/-} \\

\midrule

\multirow{4}{*}{$k_{\mathrm{nn}}$}
& $1$
& \makecell{78.84/93.99/78.31}
& \makecell{86.15/-/-} \\
& $3$
& \makecell{\textbf{83.34}/\textbf{96.03}/\textbf{81.09}}
& \makecell{\textbf{89.27}/-/-} \\
& $5$
& \makecell{83.20/95.98/80.55}
& \makecell{89.23/-/-} \\
& $7$
& \makecell{83.28/95.57/80.93}
& \makecell{89.06/-/-} \\

\midrule

\multirow{4}{*}{$\tau_{\ell}$}
& $0.001$
& \makecell{77.21/93.95/78.08}
& \makecell{86.73/-/-} \\
& $0.01$
& \makecell{\textbf{83.34}/\textbf{96.03}/\textbf{81.09}}
& \makecell{\textbf{89.27}/-/-} \\
& $0.02$
& \makecell{82.42/95.41/80.87}
& \makecell{88.54/-/-} \\
& $0.05$
& \makecell{80.17/94.25/78.99}
& \makecell{87.68/-/-} \\

\midrule

\multirow{3}{*}{$\alpha_t$}
& $1$
& \makecell{82.39/94.93/79.38}
& \makecell{87.62/-/-} \\
& $1/(t+1)$
& \makecell{\textbf{83.34}/\textbf{96.03}/\textbf{81.09}}
& \makecell{\textbf{89.27}/-/-} \\
& $1/\sqrt{t+1}$
& \makecell{83.05/95.26/80.34}
& \makecell{89.15/-/-} \\

\midrule

\multirow{5}{*}{Backbone}
& ResNet-18
& \makecell{80.82/94.98/79.05}
& \makecell{87.56/-/-} \\
& ResNet-50
& \makecell{83.02/95.98/81.05}
& \makecell{88.28/-/-} \\
& EfficientNet-B4
& \makecell{83.39/95.43/80.38}
& \makecell{89.21/-/-} \\
& EfficientNet-B5
& \makecell{\textbf{83.51}/95.41/80.86}
& \makecell{88.92/-/-} \\
& WideResNet-50
& \makecell{83.34/\textbf{96.03}/\textbf{81.09}}
& \makecell{\textbf{89.27}/-/-} \\

\bottomrule
\end{tabular}
}
\end{table}

\begin{table}[t]
\setlength{\tabcolsep}{2.0mm}{
\centering
\footnotesize
\caption{Effectiveness and scalability of $K$-means. Results are reported as the mean per-round server aggregation time $\bar T_{\mathrm{agg}}$ (s) measured over 200 communication rounds and I-AUROC/P-AUROC/PRO (\%).}
\label{tab:ablation_kmeans_scalability_texture}
\begin{tabular}{l|c|c|c}
\toprule
Configuration & $N$ & $\bar T_{\mathrm{agg}}$ & Texture \\
\midrule

\makecell{Average Aggregation}
& 5 
& \textbf{0.0241} 
& \makecell{92.76/96.10/84.42} \\

$K_\mathrm{clu}=Q,\ I_{\mathrm{clu}}=5$ 
& 5 
& 0.2373 
& \makecell{97.94/95.77/85.12} \\

$K_\mathrm{clu}=Q,\ I_{\mathrm{clu}}=10$ 
& 5 
& 0.2885 
& \makecell{99.09/97.07/86.01} \\

$K_\mathrm{clu}=Q,\ I_{\mathrm{clu}}=15$ 
& 5 
& 0.3429 
& \makecell{99.24/97.04/\textbf{86.81}} \\

$K_\mathrm{clu}=Q,\ I_{\mathrm{clu}}=20$ 
& 5 
& 0.3987 
& \makecell{99.22/97.05/86.81} \\

$K_\mathrm{clu}=Q,\ I_{\mathrm{clu}}=40$ 
& 5 
& 0.5375 
& \makecell{99.27/97.03/86.77} \\

$K_\mathrm{clu}=Q/2,\ I_{\mathrm{clu}}=15$ 
& 5 
& 0.2377 
& \makecell{98.94/96.77/86.12} \\

$K_\mathrm{clu}=2Q,\ I_{\mathrm{clu}}=15$ 
& 5 
& 0.5199 
& \makecell{\textbf{99.28}/96.96/86.45} \\

\midrule

\makecell{Average Aggregation}
& 3 
& 0.0096 
& \makecell{93.20/95.87/84.62} \\

$K_\mathrm{clu}=Q,\ I_{\mathrm{clu}}=15$ 
& 3 
& 0.1675 
& \makecell{99.13/\textbf{97.62}/85.32} \\

\midrule

\makecell{Average Aggregation}
& 10 
& 0.0705 
& \makecell{93.57/95.84/84.52} \\

$K_\mathrm{clu}=Q,\ I_{\mathrm{clu}}=15$ 
& 10 
& 0.4972 
& \makecell{99.16/96.89/85.88} \\

\bottomrule
\end{tabular}
}
\end{table}
\noindent \textbf{Impact of rEMA and Dynamic Weighting in Memory-reduce.} Table~\ref{tab:ablation_mem_reduce} presents the ablation study of rEMA and Dynamic Weighting (DW) in the memory-reduce module. All other FedDyMem network components and the server-side aggregation method are retained in all variants. When DW is removed, the local memory features are reduced by uniform averaging. Both rEMA and DW improve the performance over the variant without either component. For example, on Brain, using rEMA alone and DW alone improves the I-AUROC from $84.47\%$ to $86.75\%$ and $87.62\%$, respectively. Similar improvements are observed on MPDD and Texture. Combining rEMA and DW achieves the best performance across all datasets, reaching I-AUROC values of $89.27\%$, $83.93\%$, and $99.24\%$ on Brain, MPDD, and Texture, respectively. These results indicate that rEMA and DW are both beneficial for constructing the reduced memory bank.

\begin{table*}
\caption{Ablation study of structure for FedDyMem.}
\label{tab:alation}
\centering
\setlength{\tabcolsep}{2.0mm}{
\footnotesize
\begin{tabular}{ccc|cccccc|cccc}
\toprule
\multirow{2}{*}{+Proj}  & \multirow{2}{*}{+MG} & \multirow{2}{*}{+KM}  
& \multicolumn{6}{c|}{I-AUROC (\%)}  
& \multicolumn{4}{c}{P-AUROC (\%)} \\ 
\cmidrule{4-13} 
& & 
& Wafer & PCB & Texture
& MPDD & Brain & X-ray
& Wafer & PCB & Texture  
& MPDD \\ 
\midrule
- & - & - 
& 67.69 & 69.58 & 85.86 
& 70.41 & 69.03 & 68.92
& 70.23 & 83.23 & 85.22
& 94.19\\ 
\checkmark & - & - 
& 68.93  & 74.10  & 87.27  
& 78.28  & 74.95  & 72.71
& 74.29  & 89.26  & 90.84
& 95.77\\ 

\checkmark & \checkmark & - 
& 70.96  & 81.26  & 92.76 
& 81.05  & 85.94  & 78.12 
& 77.01  & 94.97  & 96.10
& 97.96 \\ 

\checkmark & \checkmark & \checkmark 
& \textbf{71.96}  & \textbf{83.34}  & \textbf{99.24}  
& \textbf{83.93}  & \textbf{89.27}  & \textbf{80.30}
& \textbf{78.40}  & \textbf{96.03}  & \textbf{97.04}
& \textbf{98.97}\\ 
\bottomrule
\end{tabular}
}
\end{table*}

\noindent \textbf{Impact of Parameter Initialization.} Fig.~\ref{fig:init_ablation} compares client-specific initialization (Cs-init), where different random seeds are used across clients, with identical initialization (I-init), where all clients share the same initialization. I-init consistently achieves better performance on all evaluated datasets. For example, the I-AUROC on Brain improves from $85.15\%$ to $89.27\%$. On Texture, I-init improves the I-AUROC/P-AUROC/PRO from $99.18\%$/$96.96\%$/$85.32\%$ to $99.24\%$/$97.04\%$/$86.81\%$, and on MPDD from $82.20\%$/$97.90\%$/$78.56\%$ to $83.93\%$/$98.97\%$/$79.43\%$. These results indicate that identical initialization provides a more stable starting point for constructing the initial global memory.

\noindent \textbf{Impact of Hyperparameters and Backbones.}
Table~\ref{tab:ablation_hyperparameters_backbones} evaluates the sensitivity of FedDyMem to key hyperparameters and backbone choices. Results are reported as the mean I-AUROC/P-AUROC/PRO (\%). Overall, the default settings provide the most balanced performance on PCB and Brain. Increasing the grid size from $2\times2$ to $8\times8$ improves I-AUROC from $79.02\%$ to $83.34\%$ on PCB and from $86.27\%$ to $89.27\%$ on Brain. Although $16\times16$ yields a marginally higher I-AUROC on PCB, $8\times8$ achieves better overall localization performance and the best Brain result. Similarly, $k_{\mathrm{nn}}=3$, $\tau_{\ell}=0.01$, and $\alpha_t=1/(t+1)$ achieve the best or most balanced results among the evaluated settings. For backbone choices, FedDyMem remains effective with ResNet (\cite{he2016deep}), EfficientNet (\cite{tan2019efficientnet}), and WideResNet backbones. WideResNet-50 achieves the best overall performance, with the highest Brain I-AUROC ($89.27\%$) and the best PCB P-AUROC/PRO ($96.03\%$/$81.09\%$). Although EfficientNet-B5 obtains a slightly higher PCB I-AUROC ($83.51\%$), WideResNet-50 provides the most balanced performance and is adopted as the default backbone.

\noindent \textbf{Effectiveness and Scalability of $K$-means.}
Table~\ref{tab:ablation_kmeans_scalability_texture} evaluates the effectiveness and server-side cost of $K$-means aggregation on Texture. Each client uploads $Q=H_1W_1$ reduced patch-level prototypes. Average Aggregation produces $Q$ global prototypes by averaging the $N$ local memories, whereas $K$-means clusters the received $NQ$ prototypes into $K_{\mathrm{clu}}$ global prototypes. Thus, $K_{\mathrm{clu}}$ is set relative to $Q$, and $K_{\mathrm{clu}}=Q$ keeps the same global memory capacity as Average Aggregation. With $N=5$, $K_{\mathrm{clu}}=Q$, and $I_{\mathrm{clu}}=15$, $K$-means improves I-AUROC/P-AUROC/PRO from $92.76\%$/$96.10\%$/$84.42\%$ to $99.24\%$/$97.04\%$/$86.81\%$. Increasing $I_{\mathrm{clu}}$ beyond $15$ brings only marginal changes but increases aggregation time. Using $K_{\mathrm{clu}}=Q/2$ is faster but degrades performance, while $K_{\mathrm{clu}}=2Q$ increases cost without consistent gains in P-AUROC or PRO. Therefore, $K_{\mathrm{clu}}=Q$ and $I_{\mathrm{clu}}=15$ provide a balanced trade-off. The server clusters only $NQ$ reduced prototypes per round with complexity $O(I_{\mathrm{clu}}NQK_{\mathrm{clu}}C)$. Under the default setting, the mean aggregation time increases from $0.1675$~s at $N=3$ to $0.4972$~s at $N=10$, remaining below $0.5$~s per round in the evaluated range. Compared with Average Aggregation, $K$-means adds server-side computation but consistently improves anomaly detection performance.

\noindent \textbf{Ablation Study of FedDyMem Components.} Table~\ref{tab:alation} evaluates the effects of the projection layer (Proj), memory generator (MG), and server-side $K$-means aggregation (KM). The baseline employs memory-reduce with average aggregation, without these additional components. The results show that progressively introducing Proj, MG, and KM consistently improves anomaly detection and localization performance. For example, on Wafer, the I-AUROC/P-AUROC increases from $67.69\%$/$70.23\%$ for the baseline to $68.93\%$/$74.29\%$ after adding Proj, and further to $70.96\%$/$77.01\%$ after incorporating MG. Adding KM finally improves these metrics to $71.96\%$/$78.40\%$. Similar improvements are observed on PCB, where the I-AUROC increases from $74.10\%$ with Proj to $81.26\%$ after adding MG and $83.34\%$ after further applying KM. Overall, the three components provide complementary improvements to the FedDyMem framework.

\begin{table}[t]
\centering
\setlength{\tabcolsep}{0.1mm}{
\footnotesize
\caption{Robustness to contaminated local training data. Ratio denotes the proportion of contaminated samples included in each client's local training set. Results are reported as the mean I-AUROC/P-AUROC/PRO (\%).}
\label{tab:contaminated_local_data}
\begin{tabular}{c|c|c|c|c}
\toprule
Ratio & Methods & Brain & MPDD & Texture \\
\midrule

\multirow{2}{*}{$0\%$}
& FedAvg
& \makecell{50.16/-/-}
& \makecell{61.87/91.67/55.07}
& \makecell{80.32/76.16/53.35} \\

& Ours
& \makecell{89.27/-/-}
& \makecell{83.93/98.97/79.43}
& \makecell{99.24/97.04/86.81} \\
\midrule

\multirow{2}{*}{$1\%$}
& FedAvg
& \makecell{50.14/-/-}
& \makecell{59.43/89.91/54.91}
& \makecell{77.97/73.93/49.74} \\

& Ours
& \makecell{88.90/-/-}
& \makecell{83.86/98.76/79.40}
& \makecell{98.97/96.95/85.85} \\
\midrule

\multirow{2}{*}{$5\%$}
& FedAvg
& \makecell{50.12/-/-}
& \makecell{58.49/88.07/53.24}
& \makecell{72.96/68.34/47.59} \\

& Ours
& \makecell{88.50/-/-}
& \makecell{83.52/98.47/79.55}
& \makecell{98.35/96.87/85.44} \\
\midrule

\multirow{2}{*}{$10\%$}
& FedAvg
& \makecell{50.04/-/-}
& \makecell{56.73/85.46/46.94}
& \makecell{70.97/65.00/49.46} \\

& Ours
& \makecell{88.01/-/-}
& \makecell{82.94/97.94/78.99}
& \makecell{97.98/97.19/85.05} \\

\bottomrule
\end{tabular}
}
\end{table}

\begin{figure}[t]
\centering
\includegraphics[width=\linewidth]{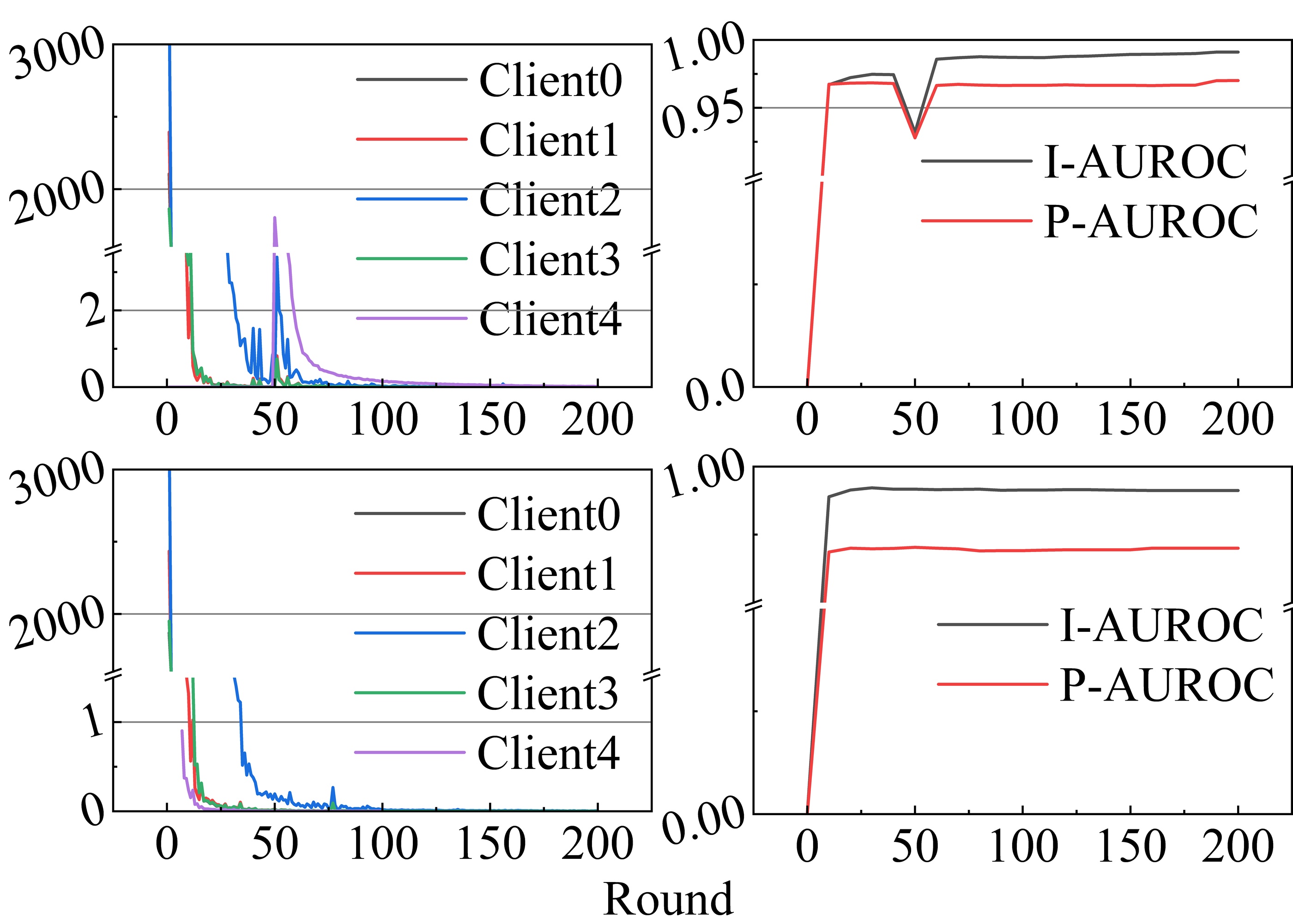}
\caption{Robustness to dynamic client participation.}
\label{fig:dynamic_participation}
\end{figure}

\begin{table}[t]
\centering
\footnotesize
\caption{Local adaptation for a newly arriving client after federated training. Results are reported as the mean I-AUROC/P-AUROC/PRO (\%).}
\label{tab:new_client_adaptation}
\begin{tabular}{l|c|c}
\toprule
Dataset & \makecell{Local-only\\(w/o $\mathbf{B}^{(T)}$)} &
\makecell{Local Adaptation\\(w/ $\mathbf{B}^{(T)}$)} \\
\midrule
Wafer
& \makecell{57.35/58.27/47.76}
& \makecell{71.66/75.96/73.11} \\

PCB
& \makecell{51.89/68.12/45.86}
& \makecell{82.88/95.67/79.96} \\

Texture
& \makecell{75.94/72.63/56.97}
& \makecell{99.11/96.80/85.91} \\

MPDD
& \makecell{54.39/87.02/51.31}
& \makecell{81.83/98.92/78.36} \\

Brain
& \makecell{50.66/-/-}
& \makecell{87.77/-/-} \\

X-ray
& \makecell{61.02/-/-}
& \makecell{79.71/-/-} \\
\bottomrule
\end{tabular}
\end{table}

\begin{table}[t]
\centering
\setlength{\tabcolsep}{1.0mm}{
\footnotesize
\caption{Performance under different client counts $N$ and Dirichlet partition parameters $\alpha_\mathrm{Dir}$.}
\label{tab:clients_partition}
\begin{tabular}{c|c|c|c|c}
\toprule
Method & $N$ & $\alpha_\mathrm{Dir}$ & Texture & MPDD \\
\midrule
FedAvg & 5 & 0.01
& \makecell{61.22/56.09/52.94}
& \makecell{55.36/86.95/50.45} \\

Ours & 5 & 0.01
& \makecell{99.12/96.85/85.30}
& \makecell{83.00/98.90/79.15} \\

FedAvg & 5 & 0.5
& \makecell{89.42/79.24/63.50}
& \makecell{63.07/91.31/67.34} \\

Ours & 5 & 0.5
& \makecell{99.14/96.85/85.48}
& \makecell{84.80/98.99/79.59} \\

FedAvg & 5 & 1.0
& \makecell{97.26/86.67/77.35}
& \makecell{63.81/92.37/70.89} \\

Ours & 5 & 1.0
& \makecell{99.13/96.87/87.07}
& \makecell{84.58/98.96/79.81} \\
\midrule
FedAvg & 3 & 0.1
& \makecell{86.72/72.02/60.83}
& \makecell{68.47/94.77/67.87} \\

Ours & 3 & 0.1
& \makecell{99.13/97.62/85.32}
& \makecell{83.24/98.91/77.77} \\

FedAvg & 5 & 0.1
& \makecell{80.32/76.16/53.35}
& \makecell{61.87/91.67/55.07} \\

Ours & 5 & 0.1
& \makecell{99.24/97.04/86.81}
& \makecell{83.93/98.97/79.43} \\

FedAvg & 10 & 0.1
& \makecell{77.94/75.18/58.13}
& \makecell{60.20/94.78/54.74} \\

Ours & 10 & 0.1
& \makecell{99.16/96.89/85.88}
& \makecell{83.30/98.99/79.08} \\

\bottomrule
\end{tabular}
}
\end{table}

\subsection{Robustness under Federated Deployment}
\label{robustness}
\noindent \textbf{Robustness to Contaminated Local Training Data.}
Table~\ref{tab:contaminated_local_data} evaluates robustness by injecting anomalous test samples into each client's normal training set. FedDyMem remains stable and consistently outperforms FedAvg across the evaluated contamination ratios. When contamination increases from $0\%$ to $10\%$, its I-AUROC only decreases from $89.27\%$ to $88.01\%$ on Brain, $83.93\%$ to $82.94\%$ on MPDD, and $99.24\%$ to $97.98\%$ on Texture. Localization is also preserved, with PRO decreasing by only $0.44\%$ on MPDD and $1.76\%$ on Texture. This stability may benefit from memory-reduce and server-side $K$-means, which summarize local features and dilute sparse client-specific contaminated features.

\noindent \textbf{Robustness to Client Participation and Data Heterogeneity.}
Fig.~\ref{fig:dynamic_participation} evaluates dynamic client participation during federated training. When Client-4 joins at round $50$, the evaluation curves show a short fluctuation and then stabilize, reaching final I-AUROC/P-AUROC values of $99.11\%$/$97.01\%$. When Client-4 drops out at round $80$, the final performance remains stable at $99.12\%$/$97.00\%$. Table~\ref{tab:new_client_adaptation} evaluates a new client arriving after federated training. Using the final global memory $\mathbf{B}^{(T)}$ as the reference memory for local adaptation consistently improves performance over local-only training. For example, PCB improves from $51.89\%$/$68.12\%$/$45.86\%$ to $82.88\%$/$95.67\%$/$79.96\%$ in I-AUROC/P-AUROC/PRO, and Brain I-AUROC improves from $50.66\%$ to $87.77\%$. These results show that the global memory provides useful shared knowledge for newly arriving clients after local adaptation. Table~\ref{tab:clients_partition} evaluates different Dirichlet parameters $\alpha_\mathrm{Dir}$ and client numbers $N$. With $N=5$ and $\alpha_\mathrm{Dir}$ varying from $0.01$ to $1.0$, FedDyMem maintains I-AUROC within $99.12\%$--$99.14\%$ on Texture and $83.00\%$--$84.80\%$ on MPDD. FedAvg is more sensitive to severe non-IID partitions, with Texture I-AUROC dropping from $97.26\%$ at $\alpha_\mathrm{Dir}=1.0$ to $61.22\%$ at $\alpha_\mathrm{Dir}=0.01$. When $N$ changes from $3$ to $10$, FedDyMem remains stable, with I-AUROC ranging from $99.13\%$ to $99.24\%$ on Texture and from $83.24\%$ to $83.93\%$ on MPDD. These results support robustness to the evaluated client participation changes, data heterogeneity, and federation scales.

\begin{table}[t]
\centering
\setlength{\tabcolsep}{0.5mm}{
\footnotesize
\caption{Robustness to a rare normal mode. Results are reported as I-AUROC/P-AUROC/PRO (\%).}
\label{tab:rare_normal_mode}
\begin{tabular}{c|c|c}
\toprule
\makecell{Injected \emph{carpet}\\ ratio $\rho$ (\# samples)}
& \makecell{\emph{Carpet}}
& \makecell{Average} \\
\midrule
$5\%$ (14)
& \makecell{96.11/93.50/80.61}
& \makecell{97.86/94.75/81.42} \\

$10\%$ (28)
& \makecell{96.91/93.76/80.80}
& \makecell{97.87/94.99/81.23} \\

$20\%$ (56)
& \makecell{97.47/96.15/81.54}
& \makecell{98.24/95.10/83.01} \\

$50\%$ (140)
& \makecell{98.79/97.42/86.01}
& \makecell{98.52/95.28/84.53} \\

$100\%$ (279)
& \makecell{\textbf{99.10}/\textbf{98.42}/\textbf{86.71}}
& \makecell{\textbf{99.03}/\textbf{96.96}/\textbf{86.73}} \\
\bottomrule
\end{tabular}
}
\end{table}

\noindent \textbf{Robustness to Rare Normal Modes.}
To examine whether server-side aggregation suppresses rare normal modes, we construct a rare-type Texture setting by removing the \emph{carpet} subtype, repartitioning the remaining data among five clients with $\alpha=0.1$, and injecting a fraction $\rho$ of \emph{carpet} training samples into one client. Table~\ref{tab:rare_normal_mode} reports performance on the rare \emph{carpet} subtype and the corresponding Texture average. With only $\rho=5\%$ ($14$ samples), FedDyMem achieves $96.11\%$/$93.50\%$/$80.61\%$ I-AUROC/P-AUROC/PRO on \emph{carpet}. When $\rho$ increases to $100\%$, the \emph{carpet} performance improves to $99.10\%$/$98.42\%$/$86.71\%$, and the Texture average increases from $97.86\%$/$94.75\%$/$81.42\%$ to $99.03\%$/$96.96\%$/$86.73\%$. These results indicate that, under the evaluated setting, rare normal subtypes can remain represented in the shared memory rather than being absorbed by dominant normal modes.

\begin{figure}[t]
\centering
\includegraphics[width=0.9\linewidth]{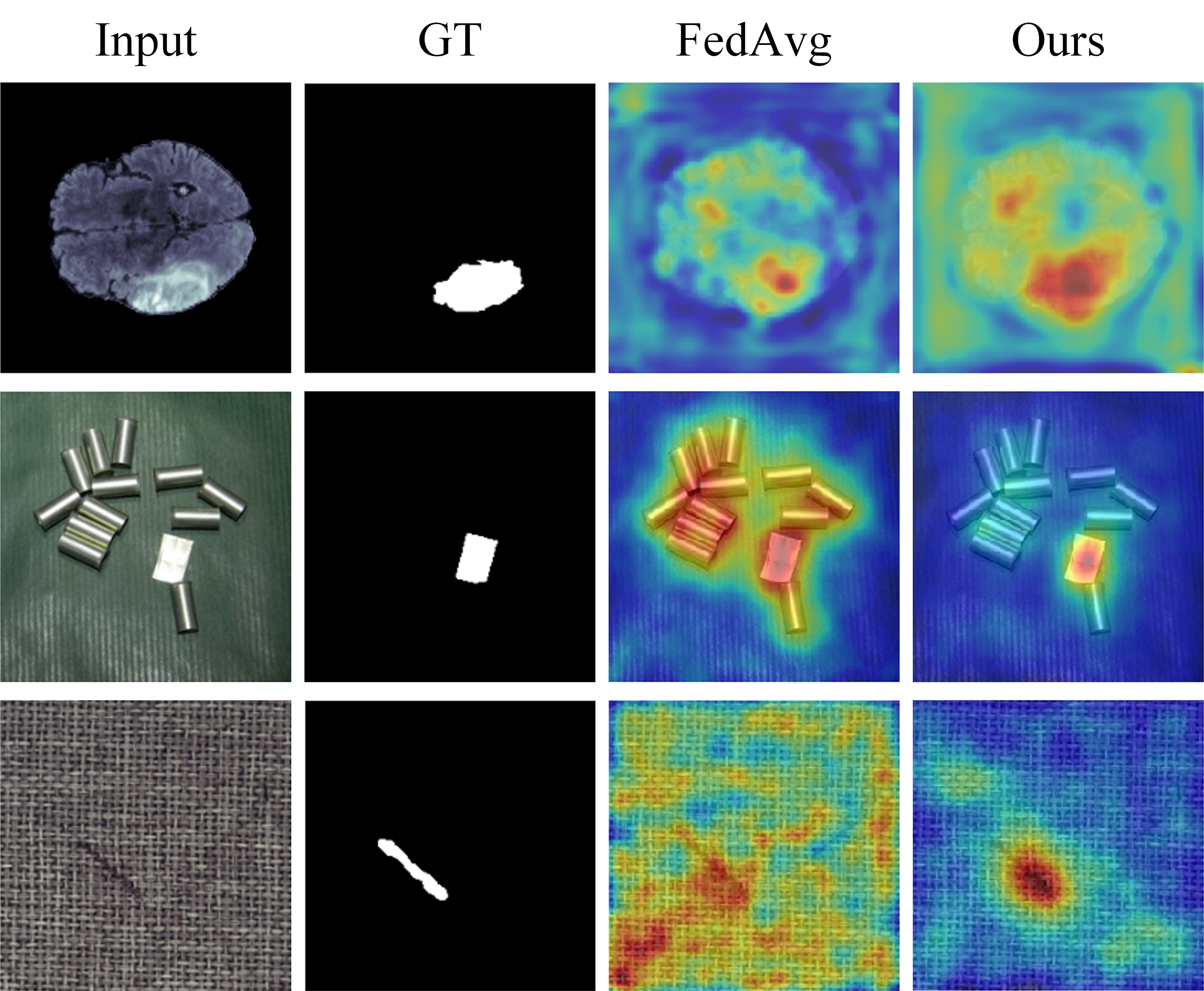} 
\caption{Qualitative comparison of FedAvg and our proposed method. The results were obtained from a randomly selected client using a global memory bank.}
\label{vis}
\end{figure}

\begin{figure*}[t]
\centering
\includegraphics[width=\linewidth]{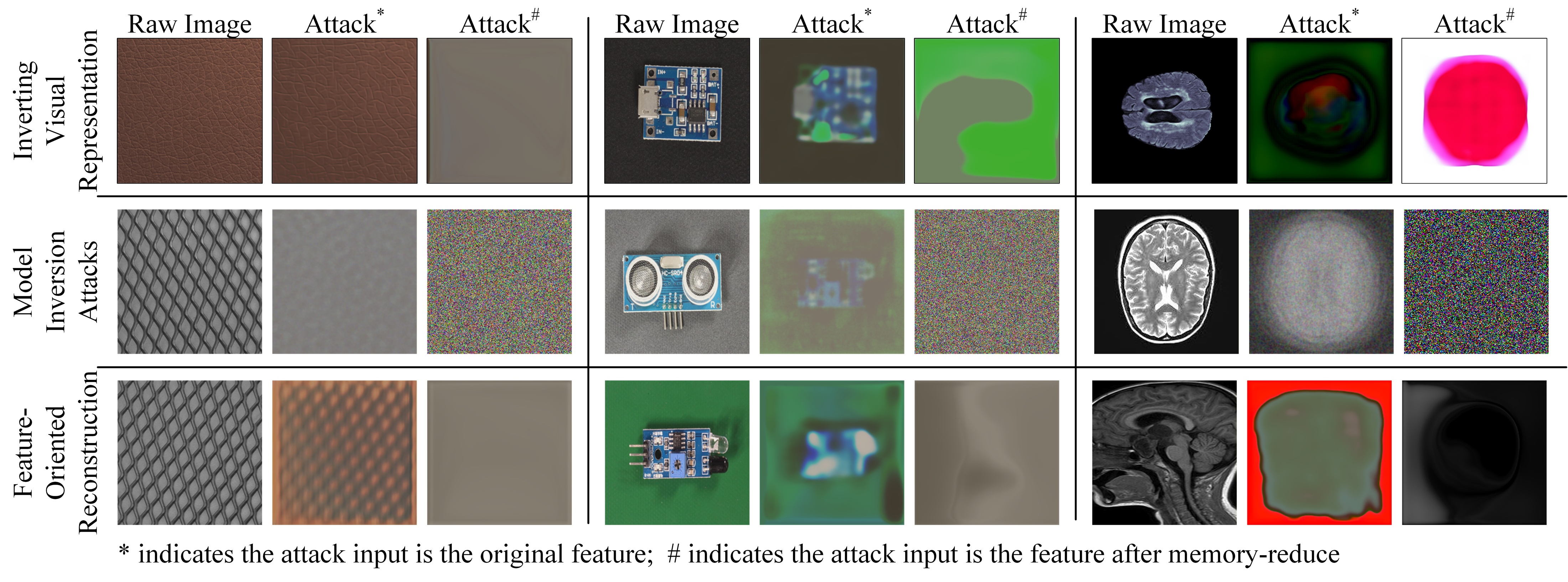} 
\caption{Qualitative reconstruction results before and after applying the proposed memory-reduce operation on Texture, PCB, and Brain.}
\label{attack}
\end{figure*}

\subsection{Qualitative Results}
Fig.~\ref{vis} provides a qualitative comparison of the prediction results between FedAvg and the proposed FedDyMem method under an identical experimental setting. The FedDyMem approach exhibits a notable improvement in prediction quality, with substantially fewer noisy regions evident in the results. Following the methodology in (\cite{lee2022cfa}), the anomaly score heatmaps are upsampled to match the spatial resolution of the input image and subsequently refined using Gaussian filtering with $\sigma=4$, ensuring smoother and more coherent boundaries. Furthermore, min-max normalization is applied to standardize the range of anomaly scores, facilitating clearer visualization and interpretation of the results.

\begin{table}[t]
\centering
\setlength{\tabcolsep}{0.4mm}{
\footnotesize
\caption{Quantitative reconstruction attack results before and after memory-reduce.}
\label{tab:attack}
\begin{tabular}{c|c|c|ccc}
\toprule
Attack 
& \multirow{2}{*}{Dataset} 
& \multicolumn{1}{c|}{+Memory-} 
& \multirow{2}{*}{SSIM$\downarrow$} 
& \multirow{2}{*}{PSNR$\downarrow$} 
& \multirow{2}{*}{FID$\uparrow$}  \\

Methods
& 
& Reduce 
& 
& 
&  \\
\midrule
\multirow{6}{*}{
    \makecell[c]{
        Deconvolution \\
        Inverting \\
        Method 
    }
} 
 & \multirow{2}{*}{Texture} 
 & -           
 & 0.6350 
 & 25.9003 
 & 246.9916 \\
 &                         
 & \checkmark 
 & 0.2088 
 & 17.4787 
 & 415.9250 \\
\cmidrule{2-6}
 & \multirow{2}{*}{PCB}   
 & -           
 & 0.5213 
 & 25.8625 
 & 281.5855 \\
 &                         
 & \checkmark 
 & 0.3379 
 & 12.1827 
 & 358.7053 \\
\cmidrule{2-6}
 & \multirow{2}{*}{Brain} 
 & -           
 & 0.2273 
 & 11.0985
 & 300.5543 \\
 &                         
 & \checkmark 
 & 0.0872 
 &  2.8752 
 & 338.1136 \\
\midrule
\multirow{6}{*}{
    \makecell[c]{
        Collaborative \\
        Inference \\
        Attack
    }
}
 & \multirow{2}{*}{Texture} 
 & -           
 & 0.0928 
 & 15.5441 
 & 300.9797 \\
 &                         
 & \checkmark 
 & 0.0179 
 & 11.4840 
 & 350.8446 \\
\cmidrule{2-6}
 & \multirow{2}{*}{PCB}   
 & -           
 & 0.4860 
 & 13.9388 
 & 313.1274 \\
 &                         
 & \checkmark 
 & 0.0178 
 & 10.7200 
 & 418.3834 \\
\cmidrule{2-6}
 & \multirow{2}{*}{Brain} 
 & -           
 & 0.1218 
 & 11.2046 
 & 365.9037 \\
 &                         
 & \checkmark 
 & 0.0038 
 &  6.5987 
 & 396.2980 \\
\midrule
\multirow{6}{*}{
    \makecell[c]{
        Feature \\
        Reconstruction \\
        Attack
    }
}
 & \multirow{2}{*}{Texture} 
 & -           
 & 0.2094 
 & 16.2157 
 & 231.2310 \\
 &                         
 & \checkmark 
 & 0.0869 
 & 15.3695 
 & 326.0596 \\
\cmidrule{2-6}
 & \multirow{2}{*}{PCB}   
 & -           
 & 0.5391 
 & 17.4453 
 & 327.5368 \\
 &                         
 & \checkmark 
 & 0.2073 
 & 14.2317 
 & 370.4590 \\
\cmidrule{2-6}
 & \multirow{2}{*}{Brain} 
 & -           
 & 0.5629 
 & 20.2312 
 & 240.1137 \\
 &                         
 & \checkmark 
 & 0.1558 
 &  7.8485 
 & 352.5879 \\
\bottomrule
\end{tabular}
}
\end{table}

\begin{table*}[t]
\centering
\footnotesize
\caption{Membership inference on the texture task under the passive honest-but-curious server threat model. Values are reported as mean $\pm$ $95\%$ normal-approximation confidence intervals over five held-out federations.}
\label{tab:membership_inference}
\begin{tabular}{l|c|c|c|c}
\toprule
\makecell{Method} &
Server-visible Message &
\makecell{ROC-AUC  $\downarrow$} &
\makecell{$\mathrm{TPR}@1\%$FPR  $\downarrow$} &
\makecell{Advantage  $\downarrow$} \\
\midrule

FedAvg
& $\Delta\boldsymbol{\theta}^{(T)}_n$
& $0.6024 \pm 0.0206$
& $0.0619 \pm 0.0104$
& $0.1687 \pm 0.0420$ \\

{Raw-memory Sharing}
& $\{\mathbf{M}^{(T)}_{n,i}\}_{i}$
& $0.7537 \pm 0.0403$
& $0.3936 \pm 0.0125$
& $0.6583 \pm 0.0428$ \\

FedDyMem
& $\mathbf{R}^{(T)}_n$
& $\mathbf{0.5310} \pm 0.0223$
& $\mathbf{0.0288} \pm 0.0122$
& $\mathbf{0.0960} \pm 0.0390$ \\

\bottomrule
\end{tabular}
\end{table*}

\begin{table*}[t]
\centering
\footnotesize
\caption{Client domain and property inference on the texture task under the passive honest but curious server threat model. Values are reported as mean $\pm$ standard deviation across five partition level attacker train and test splits constructed from 20 independently generated federation partitions.}
\label{tab:client_domain_property_inference}
\begin{tabular}{l|c|cc|ccc}
\toprule
\multirow{4}{*}{Method} &
\multirow{4}{*}{\makecell{Server-visible\\Message}} &
\multicolumn{2}{c|}{Client Domain Inference} &
\multicolumn{3}{c}{Client Property Inference} \\
\cmidrule{3-7}
&
&
\makecell{Balanced \\ Accuracy(\%) $\downarrow$} &
\makecell{Macro \\ F1(\%) $\downarrow$} &
\makecell{AUROC $\downarrow$} &
\makecell{Balanced \\ Accuracy(\%) $\downarrow$} &
\makecell{F1 Score (\%) $\downarrow$} \\
\midrule

FedAvg
& $\Delta\boldsymbol{\theta}^{(T)}_n$
& $66.8 \pm 5.9$
& $64.5 \pm 6.2$
& $0.840 \pm 0.050$
& $76.1 \pm 5.1$
& $75.0 \pm 5.5$ \\

\makecell[l]{Raw-memory Sharing}
& $\{\mathbf{M}^{(T)}_{n,i}\}_{i}$
& $94.7 \pm 2.7$
& $94.1 \pm 2.9$
& $0.985 \pm 0.010$
& $95.2 \pm 2.5$
& $95.0 \pm 2.7$ \\

FedDyMem
& $\mathbf{R}^{(T)}_n$
& $\mathbf{54.9 \pm 6.4}$
& $\mathbf{51.7 \pm 7.0}$
& $\mathbf{0.720 \pm 0.060}$
& $\mathbf{65.7 \pm 6.2}$
& $\mathbf{64.9 \pm 6.5}$ \\

\bottomrule
\end{tabular}
\end{table*}

\subsection{Privacy Verification} 
\label{sec:threat_model}

\noindent \textbf{Threat Model and Privacy Scope.}
We consider a passive honest-but-curious server that follows the protocol, records communicated messages, and performs offline inference after training. At round $t$, the server observes
\begin{equation}
\mathcal{V}_{\mathrm{srv}}
=
\left\{
\left(
\{\mathbf{R}^{(t)}_n\}_{n\in\mathcal{C}_t}
\right)
\right\}_{t=0}^{T},
\end{equation}
where $\mathcal{C}_t$ denotes the participating clients and $\mathbf{R}^{(t)}_n$ is the uploaded memory of client $n$. The attacker knows the public protocol, frozen backbone, architecture, initialization, optimization settings, and aggregation rule. It may access candidate images and auxiliary data from the same broad domain, but not the target client's raw training images or client-resident trainable parameters.

\noindent \textbf{Feature Inversion Attacks.}
We evaluate visual-content leakage by reconstructing images from server-visible memory representations using three inversion attacks, including deconvolution-based inversion \cite{dosovitskiy2016inverting}, collaborative inference attack \cite{he2019model}, and stealthy feature-oriented reconstruction attack \cite{xu2024stealthy}. Reconstruction quality is measured by SSIM, PSNR, and FID, where lower SSIM and PSNR or higher FID indicate lower visual leakage. Table~\ref{tab:attack} shows that memory-reduce consistently lowers SSIM and PSNR and increases FID across the evaluated datasets and attacks. Fig.~\ref{attack} also shows that reconstructions from reduced memories contain less recognizable structure and texture. These results indicate that memory-reduce reduces empirical visual-content leakage under the evaluated inversion attacks.

\noindent \textbf{Membership Inference Attacks.}
We next evaluate whether the server can infer whether a candidate normal image belongs to a target client's local training set:
\begin{equation}
\mathcal{A}_{\mathrm{mem}}
\left(
x,\mathcal{V}^{(T)}_n
\right)
\rightarrow
\hat{y}_{\mathrm{mem}},
\quad
y_{\mathrm{mem}}
=
\mathbb{I}
\left[
x\in\mathcal{D}^{\mathrm{train}}_n
\right].
\end{equation}
We compare three server-visible messages on the texture task, including FedAvg parameter updates $\Delta\boldsymbol{\theta}^{(T)}_n$, unreduced raw memories $\{\mathbf{M}^{(T)}_{n,i}\}_{i}$, and FedDyMem reduced memories $\mathbf{R}^{(T)}_n$. All settings use the same Deep Sets attack classifier \cite{zaheer2017deep}, with leave-one-federation-out evaluation over five independently repartitioned shadow federations. As shown in Table~\ref{tab:membership_inference}, raw memory sharing yields the strongest membership leakage, with ROC-AUC of $0.7537$, $\mathrm{TPR}@1\%\mathrm{FPR}$ of $0.3936$, and advantage of $0.6583$. FedAvg also exposes membership information, with ROC-AUC of $0.6024$. FedDyMem obtains the lowest attack performance, with ROC-AUC of $0.5310$, $\mathrm{TPR}@1\%\mathrm{FPR}$ of $0.0288$, and advantage of $0.0960$, indicating reduced empirical membership leakage.

\noindent \textbf{Client Domain and Property Inference Attacks.}
We further evaluate whether server-visible messages reveal client-level distributional information. On the texture task, let $p_{n,g}$ denote the proportion of subtype $g$ in client $n$'s training data. The dominant client domain and binary client property are defined as
\begin{equation}
y^{\mathrm{dom}}_n
=
\arg\max_{g} p_{n,g},
\qquad
z_{n,g}
=
\mathbb{I}
\left[
p_{n,g}\geq\rho
\right].
\end{equation}
We compare FedAvg updates, raw memory messages, and FedDyMem reduced memories using the same tokenization and Deep Sets multitask attacker \cite{zaheer2017deep}. Table~\ref{tab:client_domain_property_inference} shows that raw memory sharing exposes the strongest client distribution signal, with $94.7\%$ domain Balanced Accuracy and property AUROC of $0.985$. FedAvg also reveals client-specific information, with $66.8\%$ domain Balanced Accuracy and property AUROC of $0.840$. FedDyMem achieves the lowest attack scores, reducing domain Balanced Accuracy to $54.9\%$ and property AUROC to $0.720$. These results show that memory-reduce lowers empirical client domain and property leakage relative to raw feature sharing and parameter update sharing under the evaluated threat model.

\section{Conclusion}

We proposed FedDyMem, a federated framework for unsupervised image anomaly detection that exchanges reduced dynamic memory banks instead of full model parameters. FedDyMem aligns heterogeneous client-side normal representations through a memory generator and metric loss, compresses local memories through memory-reduce, and aggregates multiple normal modes using server-side $K$-means. Experiments on six industrial and medical federated UAD tasks show that FedDyMem improves anomaly detection and localization over the evaluated federated baselines while reducing the size of the exchanged representation. The additional studies further support its robustness under the evaluated non-IID partitions, client participation changes, limited training-data contamination, and rare normal modes.

\noindent \textbf{Limitations.} FedDyMem is evaluated on federations where clients share the same inspection or screening objective but have heterogeneous normal data across sources or subtypes. Its applicability to semantically different anomaly detection tasks remains unexplored. FedDyMem is specifically designed for UAD architectures with trainable memory-generation modules. Applying it to other UAD architectures would require corresponding changes to the representation-sharing and aggregation mechanisms. Although we evaluate varying client numbers and non-IID partitions, larger-scale deployments and coordinated adversarial behaviors may pose additional challenges. Memory-reduce mitigates privacy risks through representation compression, while stronger protection against broader attack models may require secure aggregation or differential privacy.

\backmatter

\bibliography{ref}

\end{document}